\documentclass{statsoc}
\RequirePackage[OT1]{fontenc}
\usepackage[top=1in, bottom=1in, left=0.5in, right=0.5in]{geometry}
\usepackage{amsmath}
\usepackage{amsfonts}
\usepackage{amssymb}
\usepackage{hyperref}
\usepackage{graphicx}
\usepackage{cleveref}
\usepackage{float}
\usepackage{enumerate}
\usepackage[mathscr]{euscript}
\usepackage{multirow}
\usepackage{xcolor}
\usepackage[font=small,labelfont=bf]{caption}
\usepackage[parfill]{parskip}
\usepackage{mathtools}
\usepackage{mathrsfs}
\usepackage{lscape}
\usepackage{tikz}
\usepackage[]{algorithm2e}
\usetikzlibrary{arrows,shapes,trees}
\usepackage{multirow}
\usepackage{epstopdf}
\usepackage{soul}
\usepackage{booktabs}
\usepackage{setspace}
\usepackage{color}
\usepackage{indentfirst} 
\usepackage{subcaption}
\usepackage{mwe}
\newtheorem{Thm}{Theorem}[section]
\newtheorem{Assumption}{Assumption}
\newtheorem{Lem}[Thm]{Lemma}

\newcommand*\X{\mathbf{X}}
\newcommand*\XX{\X(-u)}

\newcommand*\D{\mathbf{D}}

\newcommand*\Y{\mathbf{Y}}

\newcommand*\x{\mathbf{x}}

\newcommand*\PP{\mathcal{P}}
\newcommand*\E{\mathbb{E}}
\newcommand{\argmin}{\operatornamewithlimits{arg\,min}}

\newcommand\blfootnote[1]{%
  \begingroup
  \renewcommand\thefootnote{}\footnote{#1}%
  \addtocounter{footnote}{-1}%
  \endgroup
}
\title[Dynamic Networks with Multi-scale Temporal Structure]{Dynamic Networks with Multi-scale Temporal Structure}
\author[Author 1 {\it et al.}]{Xinyu Kang}
\address{Boston University,
Boston,
USA}
\author{Apratim Ganguly}
\address{Boston University,
Boston,
USA}
\author{Eric D. Kolaczyk}
\address{Boston University,
Boston,
USA\blfootnote{\textit{Address for correspondence:} Eric Kolaczyk, Department of Mathematics \& Statistics, Boston University, 111 Cummington Mall, Boston MA, 02215, USA. E-mail: kolaczyk@bu.edu}\\}
\begin{document}
\begin{abstract}
We describe a novel method for modeling non-stationary multivariate time series, with time-varying conditional dependencies represented through dynamic networks. Our proposed approach combines traditional multi-scale modeling and network based neighborhood selection, aiming at capturing temporally local structure in the data while maintaining  sparsity of the potential interactions. Our multi-scale framework is based on recursive dyadic partitioning, which recursively partitions the temporal axis into finer intervals and allows us to detect local network structural changes at varying temporal resolutions. The dynamic neighborhood selection is achieved through penalized likelihood estimation, where the penalty seeks to limit the number of neighbors used to model the data. We present theoretical and numerical results describing the performance of our method, which is motivated and illustrated using task-based magnetoencephalography (MEG) data in neuroscience.\par
\end{abstract}
\section{Introduction}
\noindent
The automated, simultaneous monitoring of each unit in a large complex system has become commonplace. Frequently the data observed in such a system is in the form of a high dimensional multivariate time series. Domain areas where such a paradigm is particularly pertinent include computational neuroscience (e.g., temporal imaging across voxels or brain regions) and finance (e.g., investment returns across stocks or levels of lending among central banks).  The combination of system and time series in these settings suggests a role for dynamic network modeling, a quickly developing area of study in the field of network analysis.

As the basic object of treatment in this paper we consider a multivariate time series, $\left(X_t(1),\cdots, X_t(N)\right)$, observed at each of $N$ units at times $t=1,\ldots, T$, as a set of measurements from across a system.  We will use a graph $G=(V,E)$ to describe the conditional dependencies among the time series across the system. Here $V=\{1,\ldots,N\}$ are vertices corresponding to the $N$ units in the system, and $E$ is the collection of vertex pairs joined by edges.   Given data, we seek to select an appropriate choice of $G$ that best characterizes the system, using techniques of statistical modeling and inference. This task is known as network topology inference~\cite[Ch 7.3]{Kolaczyk:2009:SAN:1593430}.  The notion of association used in this paper is a type of partial correlation, analagous to that underlying so-called Granger causality (\cite{granger1969investigating}).  Granger causal types of models have been widely utilized in financial economics -- see \cite{hamilton1983oil}, \cite{hiemstra1994testing} and \cite{sims1972money}, for example -- and in biological studies -- see \cite{mukhopadhyay2007causality}, \cite{bullmore2009complex} for instance.

Granger causal models traditionally assume a stationary time series and take a vector-autoregressive (VAR) form.  Here we adopt a restricted-VAR(p) model, defined as a VAR model without the self driven components:
\begin{align*}
X_t(u) = \sum_{v \in V\backslash\{u\}}\sum_{\ell=1}^p X_{t-\ell}(v)\theta^{(\ell)}(u,v) + \epsilon_t(u),
\end{align*}
where $\theta^{(\ell)}(u,v)$ collects the influence of the node $v$ on node $u$ at lag $\ell$ and $\epsilon_t(u)$ is independent Gaussian white noise.  It is said that $X(v)$ Granger causes $X(u)$ if and only if $\theta^{(\ell)}(u,v) \neq 0$ for some $\ell = 1,\cdots, p$. We use the term `restricted' in describing this model because we restrict $\theta^{(\ell)}(u,u)$ to be $0$ for all $u$, $\ell$.  This requirement is made for notational convenience, and without loss of generality, in that it essentially assumes the self-driven component has been removed and that our network characterizes only relationships between distinct nodes.  The notion of `network' in this framework is made precise through graphs defined as a function of the underlying graphical model. That is, through conditional independence relations, coded in one-to-one correspondence with patterns of non-zero elements among the
$\theta^{(\ell)}(u,v)$.  Specifically, $G = (V,E)$ is a directed graph with an edge from $v$ to $u$ if and only if $\|\boldsymbol\theta(u,v)\|_2 \neq 0$, where $\boldsymbol\theta(u,v) = \left(\theta^{(1)}(u,v),\cdots,\theta^{(p)}(u,v) \right)^\prime$.

Multivariate time series data is often non-stationary.  Furthermore, it is not uncommon to expect changes in a system across multiple time scales.  For example, it is widely recognized that financial time series of quantities like equity, interest, and credit can exhibit volatility across multiple scales (e.g., \cite{fouque2011multiscale}).  Similarly, it is believed that neuronal dynamics within the cerebral cortex in the brain interact with anatomical connectivity in such a way as to produce functional connectivity relationships between brain regions at multiple time scales (\cite{honey2007network}). These observations suggest the need for a notion of multi-scale analysis when doing network-based modeling of multivariate time series in systems like these.  However, while temporal multi-scale analysis is a concept well-established in time series analysis, it does not appear to have yet emerged in network modeling.

Motivated by the elements of the above discussion, we focus in this paper on the problem of detecting dynamic connectivity changes across multiple time scales in a network-centric representation of a system, based on multivariate time series observations. Our approach combines the traditional Granger causal type of modeling with partition-based multi-scale modeling.  We adopt a change point perspective, so that our model class consists of concatenations of restricted-VAR(p) models, each with its own $\boldsymbol\theta$ constant over a given interval of time. The result is then a time-indexed directed graphical model, from which we define a dynamic network $G_t = (V, E_t)$, in analogy to the stationary case. Our goal is then to infer the change points distinguishing the stationary intervals and the corresponding edge sets $E_t$.

A number of works in recent years have focused on modeling multivariate time series using causal network types of models. A common theme among these is to generalize the work of \cite{meinshausen2006high}, who show that the Lasso can consistently recover the neighborhood structure of a Gaussian graphical model in high-dimensional settings under appropriate assumptions.  Seminal examples of such extensions include~\cite{bolstad2011causal}, where they assume the time series are stationary and carry out variable selection using group-lasso principles; and \cite{JMLR:v16:basu15a}, where they estimate the network Granger causality for panel data using the group-lasso. Similarly, in the work by \cite{barigozzi2014nets}, networks are defined and inferred through use of the long-run partial correlation matrix between multiple time series. For non-stationary multivariate time series processes,  \cite{long2005nonstationary} use time-varying auto-regressive models with adaptively chosen -- but fixed -- windows.  These latter are applied to functional MRI data.

While we make use of ideas similar to those above, our approach is significantly different from those proposed previously in the sense that we incorporate them within a multi-scale framework. Multi-resolution analysis was formally proposed by \cite{mallat1989theory} and others in the late 1980's and has been known for mathematically elegant, computationally efficient and often domain-specific representations of data that are inhomogeneous in their support.  While there is by now a vast literature on the topic of multiscale statistical modeling, with literally scores of representations for standard signal and image analysis applications alone, a key representation is that of recursive dyadic partitioning.  A fundamental result from \cite{donoho1997cart} relates the method of recursive dyadic partitioning and the selection of a best-orthonormal basis, where the basis is selected from a class of unbalanced Haar wavelets. The partition-based multi-scale method has proven to be particularly natural and useful in extending wavelet-like ideas to nontraditional settings, for example, in the context of generalized linear models, irregular spatial domains, etc. -- see \cite{kolaczyk2005multiscale}, \cite{louie2006multiscale}, and \cite{willett2007multiscale}, for instance.  For a recent survey of statistical methods for network inference from time series, in general, see~\cite[Sec 4.2]{betancourt2017bayesian}.

Our main contribution in this paper is to present a partition-based multi-scale dynamic causal network model, and a corresponding method of network topology inference, that captures the dynamics of a system in a manner sensitive to changes at multiple time scales, while encouraging sparsity of network connectivity. There are three key elements in the framework: (i) we partition the non-stationary time axis into blocks at various scales, with independent, stationary VAR models indexed by blocks; (ii) to prevent overfitting, we impose a counting penalty to penalize the number of blocks used; and (iii) we do neighborhood selection within each block using a group-lasso type of estimator.\par
This paper is organized as follows. In Section 2, we provide the details our partition-based dynamic multi-scale network model and methodology. In Section 3, we present several characterizations of theoretical properties of our estimator. The broad potential impact of our method is demonstrated in Section 4, through the use of both simulated data and a magnetoencephalography (MEG) data set. Technical proofs are provided in the appendix. Code implementing the methodology proposed in this paper is available from \url{https://github.com/KolaczykResearch/MS-Dyn-Networks-Code}.

\section{Partition-based multi-scale dynamic network models}
\noindent
In this section we define the class of dynamic network models developed in this paper, we describe our proposed approach to network inference within this class, and we summarize the implementation of this approach in the form of an algorithm.
\subsection{Piecewise vector autoregressive models}
\noindent
We are interested in non-stationary multivariate time series, as the stationarity assumption required by traditional vector autoregressive modeling is overly restrictive in the types of financial and biological applications motivating our work.  Accordingly, we define a class of restricted piece-wise vector autoregressive models.  These models are of order $p$ [rP-VAR(p)] and break the non-stationary time series into an unknown number of $M$ stationary blocks, with a stationary restricted VAR(p) model within each block.

More specifically, we equip the parameters in our previously defined restricted VAR(p) model with a time index:  
\begin{equation}
X_{t}(u) = \sum_{v\in V\backslash\{u\}}\sum_{\ell=1}^p X_{t-\ell}(v)\, \theta^{(\ell)}_{t}(u,v) + \epsilon_{t}(u) \enskip .
\label{eq:time.ind.rVarp}
\end{equation}
Next we restrict the coefficient vectors $\boldsymbol\theta_t(u,v) = \left(\theta^{(1)}_t(u,v),\cdots,\theta^{(p)}_t(u,v) \right)^\prime$ to be constant within each of $M$ blocks defined by change points $\tau_0 = 0$ and $\tau_{M+1} = T$.  Finally, we assume independence of the multivariate time series across blocks.  We then capture the evolving dependency structure of the data using a time-varying directed graph $G = (V, E_t)$ with an edge from $v \rightarrow u$ if and only if $\|\boldsymbol\theta_t(u, v)\|_2 \neq 0$. 

Certain of these choices could be relaxed, at the expense of a nontrivial increase in complexity of both computation and exposition.  The assumption of independence between blocks could be relaxed to allow for weak dependence over $p$ time steps just prior to and after each changepoint, following the suggestion in~\cite[Remark 1]{davis2008break}.  Additionally, we assume the number of lags $p$ is fixed and known.  In contrast, an unknown value of $p$ in principle could be incorporated into our framework, with selection made through an additional penalty term. 

To organize the collection of blocks defining our class of rP-VAR(p) models, we use the notion of recursive partitioning.  This choice is both consistent with our goal of capturing multi-scale structure (as described above) and facilitates the development of sensible algorithms for computational purposes.  We will consider two types of partitioning: recursive dyadic partitioning and (general) recursive partitioning.  Without loss of generality, we consider partitioning restricted to the unit interval $(0, 1]$ interchangeably with partitioning of the interval $(0,T]$.  A partition $\PP$ of $(0,1]$ is a decomposition of the latter into a collection of disjoint subintervals whose union is the unit interval.  In our treatment we restrict attention to partitions of finite cardinality.

Both recursive dyadic partitioning and recursive partitioning produce partitions $\PP$ by recursively partitioning the unit interval.  They differ only in the rule defining the choice of partitions that may be produced at each iteration, with that for the former being more restrictive than that for the latter.  Under recursive dyadic partitioning, starting with the unit interval, we recursively split some previously resulting interval into two sub-intervals of equal length.  Under recursive partitioning more generally, the restriction to dyadic subintervals is removed.  Under both approaches, partitioning is done only up to the resolution of the data.  Therefore, with $T$ observation times, partitioning is done only at the points $\{i/T\}_{i=1}^{T-1}$, and only up to a total of $T$ subintervals.  Under recursive dyadic partitioning, we require that the number of observations $T = 2^J$ be a power of two.

Let $\PP^*_{D_y}$ denote the complete recursive dyadic partition (with the dependence on $T$ suppressed for notational convenience), and $\PP^*$, a complete recursive partition.  Additionally, denote by $\PP\preceq \PP^*_{D_y}$ (respectively, $\PP\preceq \PP^*$) a subpartition of $\PP^*_{D_y}$ (respectively, $\PP^*$), i.e., as one of the partitions defined through the process of successive refinement from $(0,1]$ to $\PP^*_{D_y}$ (respectively, $\PP^*$).  This notation helps emphasize one of the key advantages of the partition-based perspective, i.e., that algorithms to search efficiently over model spaces indexed by these partition classes can be designed to do so in $\mathcal{O}(T)$ and $\mathcal{O}(T^3)$ computational complexity, respectively, using dynamic programming principles.  See~\cite{kolaczyk2005multiscale}.  The advantage of recursive dyadic partitioning over recursive partitioning therefore typically is in computational cost.  We will define a class of rp-VAR(p) models indexed by these partition classes and propose algorithms for model selection that exploit the accompanying dynamic programming principles.

\subsection{Network Inference}
\noindent
The graphs $G$ corresponding to the restricted piece-wise $VAR(p)$ class of models we have introduced can be thought of as a union of the neighborhoods surrounding each node $u$.  And, in fact, we will infer the topology of the network $G$ neighborhood by neighborhood. 

Consider, for example, the cartoon illustration in Figure~\ref{fig:cartoon} where, without loss of generality, the focus is on the local neighborhood of a node/series $u$ and $T=160$ for illustration. From time $[0, 60)$, each of the four other nodes $B, D, C,$ and $E$ Granger causes $u$. From time $[60, 80)$, only node $B$ Granger causes $u$, and for the rest of the time, $B$ and $D$ Granger cause $u$.  Under our proposed approach, we estimate the times $\tau_m$ at which the changes happened. Given the estimated change points, we then infer the neighborhood structure during the time interval $[0,\hat{\tau_1})$, and then $[\hat\tau_1, \hat\tau_2)$, and so on. Put simply, our approach is to estimate the change-points and the neighborhood structures within each stationary time-interval defined by those change-points, where the changepoints are defined through either a recursive dyadic partition or a recursive partition.  We describe each of these two cases in turn below.
\begin{figure}[!htb]
\centering
\begin{tikzpicture}[scale=.6]
\draw [<->] (0,0) -- (6,0)
node[pos=0,above] {$\tau_0=0$};
\draw [<->] (6,0) -- (8,0)
node[pos=0,above] {$\tau_1=60$};
\draw [<->] (8,0) -- (16,0)
node[pos=0,above right] {$\tau_2=80$}
node[pos=1,above] {$\tau_{max}=160$};
\end{tikzpicture}

\rule{0ex}{10ex}
\hspace{-0.5in}
\begin{tikzpicture}[scale=.5, transform shape]
\tikzstyle{every node} = [circle, fill=gray!30]
\node (a) at (0, 0) {u};
\node (b) at +(0: 1.5) {B};
\node (c) at +(60: 1.5) {C};
\node (d) at +(120: 1.5) {D};
\node (e) at +(180: 1.5) {E};
\foreach \from/\to in {b/a, c/a, d/a, e/a}
\draw [->] (\from) -- (\to);
\end{tikzpicture}
\hspace{0.2in}
\begin{tikzpicture}[scale=.5, transform shape]
\tikzstyle{every node} = [circle, fill=gray!30]
\node (a) at (0, 0) {u};
\node (b) at +(0: 1.5) {B};
\node (c) at +(60: 1.5) {C};
\node (d) at +(120: 1.5) {D};
\node (e) at +(180: 1.5) {E};
\foreach \from/\to in {b/a}
\draw [->] (\from) -- (\to);
\end{tikzpicture}
\hspace{0.3in}
\begin{tikzpicture}[scale=.5, transform shape]
\tikzstyle{every node} = [circle, fill=gray!30]
\node (a) at (0, 0) {u};
\node (b) at +(0: 1.5) {B};
\node (c) at +(60: 1.5) {C};
\node (d) at +(120: 1.5) {D};
\node (e) at +(180: 1.5) {E};
\foreach \from/\to in {b/a, d/a}
\draw [->] (\from) -- (\to);
\end{tikzpicture}\\[5pt]
\caption{Cartoon version of the underlying network structure.}
\label{fig:cartoon}
\end{figure}
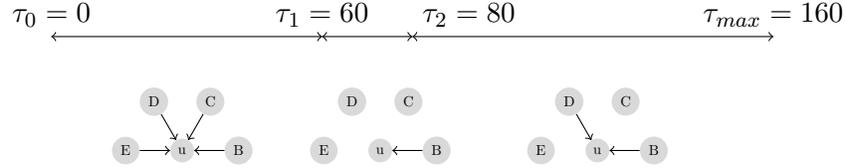

%\subsubsection{Multi-scale modeling and inference using RDP}
%\noindent
Suppose that our changepoints $\tau_i$ are restricted to correspond to the boundaries of some recursive dyadic partition.  For a given node $u$, we estimate the vector $\boldsymbol\theta\equiv \left( \theta_{t_i}^{(\ell)}(u,v)\right)$, defined for all nodes $v\in V\setminus \{u\}$ and at all times $t_i=i/T$ where $i=1,\ldots, T$, by choosing some optimal member from the classes rP-VAR(p) defined by all possible partitions $\PP\preceq \PP^*_{D_y}$ of the unit interval. 
Formally, we define the space of all possible values of $\boldsymbol{\theta}$
\small
\begin{align}
\Gamma_{RDP}^{(N-1)p} \equiv 
&\left\{\boldsymbol{\theta} \left | \theta_t^{(\ell)}(u, v) = \beta_0^{(\ell)}(u,v) + \sum_{I \in \ell_{NT}(\PP)}\beta_I^{(\ell)}(u,v) h_I(t) \right. 
 \quad\forall \, \ell, v, \text{for some }  \PP \preceq \PP_{D_y}^*  \right\}\enskip ,
\label{Space_RDP}
\end{align}
\normalsize\noindent
where $\PP$ is a partition common to all coefficient functions $\theta^{(\ell)}_t(u,v)$ across nodes $v$ and lags $\ell$, for each fixed $u$. In this expression, $\ell_{NT}(\PP)$ is the set of all non-terminal (NT) intervals encountered in the construction of $\PP$, while $\beta_0^{(\ell)}(u,v)$ and $ \beta_I^{(\ell)}(u, v)$ are the (non-zero) coefficients in a reparameterization of $\theta_t^{(\ell)}(u, v)$ with respect to the unique (dyadic) Haar wavelet basis  $\{h_I\}_{I \in \ell_{NT}(\PP^*_{Dy})}$ associated with the complete recursive dyadic partition $\mathcal{P}_{Dy}$.  In particular, a wavelet $h_I$ has as its support the interval $I$, and is proportional to the values $1$ and $-1$ on the two subintervals defined by a split at the midpoint of $I$.  See~\cite{donoho1997cart} or~\cite{kolaczyk2005multiscale}, for example, for details on this correspondence between recursive dyadic partitions and classical Haar wavelet bases.  It is this correspondence that makes explicit the multiscale nature of our approach.

Based on this model class, we define a complexity-penalized estimator $\boldsymbol{\hat\theta}_{RDP}$ of $\boldsymbol{\theta}$ as follows:
\begin{align}
\hat{\boldsymbol{\theta}}_{RDP} \equiv \argmin_{\boldsymbol{\tilde\theta} \in \Gamma_{RDP}^{(N-1)p}}\left\{ -\log p\left(\X(u)|\X(-u),\boldsymbol{\tilde\theta}\right) + 2\sum_{v \in V \backslash \{u\}}\text{Pen}_{RDP}(\boldsymbol{\tilde\theta}(u,v))\right\}\enskip .
\label{rdp}
\end{align}
Here $\X(-u)$ is the lagged design matrix of dimension $T \times (N-1)p$ based on the observed time series information for all nodes except $u$.  That is, we define $\X(-u) = (\X(1), \cdots, X(u-1),X(u+1),\cdots, \X(N))$, with each $\X(\cdot)$ a $T \times p$ matrix defined as $\X(\cdot) = (\X_{-1}(\cdot),\cdots,\X_{-p}(\cdot))$, where $\X_{-\ell}(\cdot)$ contains the lagged observations $\X_{-\ell}(\cdot) = (X_{T-\ell}(\cdot),\cdots,X_{-\ell+1}(\cdot))^\prime$. The function $\text{Pen}_{RDP}(\boldsymbol{\tilde\theta}(u,v))$ is the penalty imposed for incorporating node $v$ into the model.

%\subsubsection{Multi-scale modeling and inference using RP}
%\noindent
Now consider the case where the network changepoints $\tau_i$ are restricted to correspond to the boundaries of some arbitrary (i.e., non-dyadic) recursive partition.  Define $\mathcal{L}$ to be the library of all $(T-1)!$ possible complete recursive partitions $\PP^*$, and let
\small
\begin{align}
& \Gamma_{RP}^{(N-1)p} \equiv 
& \left\{\boldsymbol{\theta} \left| \theta_t^{(\ell)}(u, v) = \beta_0^{(\ell)}(u,v) + \sum_{I \in \ell_{NT}(\PP)}\beta_I^{(\ell)}(u,v) h_I(t) \,\, \forall \ell,v, \text{ for some } \PP \preceq \PP^*, \PP^* \in \mathcal{L} \right.\right\} \enskip .
\label{Space_RP}
\end{align}
\normalsize
Here $\{h_I\}_{I \in \ell_{NT}(\PP^*)}$ is the unique (unbalanced) Haar wavelet basis corresponding to a given complete recursive partition $\PP$.  As in the case of the classical dyadic Haar basis, there will be $T$ piecewise constant basis functions for $T$ time points, each indexed according to its support interval $I$ and proportional in value to $1$ or $-1$ on two subintervals (except for one `father' wavelet, defined to capture the average of $\theta_t^{(\ell)}(u,v)$ over $(0,T]$).  But, unlike before, the subintervals defining these wavelets are not necessarily of equal length.  This definition allows, for example, for the representation of non-dyadic changepoints in a  potentially more efficient manner (i.e., using fewer recursive splits).  See~\cite{kolaczyk2005multiscale} for details.

Analogous to the dyadic case, our estimator defined under recursive partitioning is given by:
\begin{align}
\hat{\boldsymbol{\theta}}_{RP} \equiv \argmin_{\boldsymbol{\tilde\theta} \in \Gamma_{RP}^{(N-1)p}}\left\{ -\log p\left(\X(u)|\X(-u),\boldsymbol{\tilde\theta}\right) + 2\sum_{v \in V\backslash \{u\}} \text{Pen}_{RP}(\boldsymbol{\tilde\theta}(u,v))\right\} \enskip .
\label{rp}
\end{align}
This is a maximum complexity-penalized likelihood estimator of $\boldsymbol{\theta}$ defined on a much broader space. It includes all possible partitions that divide the unit interval into $M\leq T$ blocks, where sub-intervals need not necessarily be of equal size. This increase in richness of representation, however, will be seen to come at a computational cost.

%\subsubsection{Penalty functions for estimator using RDP and RP}
%\noindent
The penalty function used to define these two estimators is described as follows.  Define the $p$-length vector $\boldsymbol{\theta}_I(u,v)$ to be the collection of (fixed) values $\theta^{(\ell)}_t(u,v)$ over all lags $\ell=1,\ldots, p$ for $t\in I$.  For recursive partitioning, we then define the penalty of incorporating a given node $v$ into the model to be
\begin{align}
\text{Pen}_{RP}(\boldsymbol{\theta}(u, v)) = \frac{3}{2}\#\{\mathcal{P}(\boldsymbol{\theta})\} \log T + \lambda\sum_{I \in \mathcal{P}(\boldsymbol{\theta})} \|\boldsymbol{\theta}_I(u, v)\|_2 \enskip .
\label{penaltyFunction}
\end{align}
For recursive dyadic partitioning, we replace the value $3/2$ by $1/2$, indicating that we penalize less severely in the simpler model class.

Note that this penalty is composed of two parts. In the first part, $\#\{\PP(\boldsymbol{\theta})\}$ is the cardinality of the partition $\PP(\boldsymbol{\theta})$ corresponding to a given value $\boldsymbol{\theta}$ in $\Gamma^{(N-1)p}_{RDP}$ or $\Gamma^{(N-1)p}_{RP}$.  Because this partition is assumed common across lags $\ell$ and for all $v\in V\setminus \{u\}$, it may be thought of as a union, i.e., 
$\PP(\boldsymbol{\theta}) = \bigcup\limits_v \PP(\boldsymbol{\theta}(u,v))$, where $\PP(\boldsymbol{\theta}(u,v))$ is a partition corresponding specifically to the dynamic behavior of the coefficients $\theta^{(\ell)}_t(u,v)$ collectively over all lags $\ell$.  Thus the contribution of $\#\{\PP(\boldsymbol{\theta})\}$ to the penalty may be thought of as counting the number of times there is a need to insert a changepoint due to a change in the relation of node $u$ with any other node $v$ at any lag $\ell$.  That is, it controls the number of partitions for the entire neighborhood.  

The second part of the penalty in (\ref{penaltyFunction}) is a sum, over intervals $I$ in the relevant partition $\PP$, of the $\ell_2$ norms of the corresponding coefficient lag vectors. It is essentially a group lasso type penalty, in the spirit of that originally proposed by \cite{yuan2006model}, with tuning parameter $\lambda$. The purpose of introducing this term is to encourage sparseness in the connectivity of each neighborhood, and hence of the network as a whole.  Our use of the group lasso here derives from the definition of our network $G$, where an edge is present regardless of in which lag there is a causal effect of a node $v$ on the node $u$. 
The choice of tuning parameter controls the amount of shrinkage of the group of coefficients. Large $\lambda$  results in sparser coefficient vectors.   We describe a method for choosing the tuning parameter in Section~3.

\subsection{Implementation}
\noindent
In this section, we discuss the implementation of our proposed methods of inference. For both the recursive dyadic partitioning estimator in (\ref{rdp}) and the recursive partitioning estimator in (\ref{rp}), the general structure of the algorithm is similar.  We describe the latter here and, for the sake of completeness, provide the former in the appendix.

\begin{algorithm}[bht]
 \KwData{$\X(u)$, $\X(-u)$, $p$}
 \KwResult{$\boldsymbol{\hat\theta}_{RP}$}
 \For{j = 1:p}{
 	\For{i = 1: T-j+1}{
  		Compute and store $pl_I$ on each interval $I$ using:\\
  		$pl_I = \sum_I(\X_I(u))^2$ for $I = \{t: t \in [i, i+j)\}$\;
  		optimumModel $\gets pl_I$\;
  }
 }
  \For{j = p+1:T}{
 	\For{i = 1: T-j+1}{
  		Fit restricted VAR(p) model for $\X_I(u), I = \{t: t \in [i, i+j)\}$\;
  		Compute and store $pl_I$ on each interval $I$\;
  		\eIf{$pl_I \leq pl_{I_l^i} + pl_{I_r^i} + \text{Penalty}$}{
  		 optimumModel $\gets pl_I$\;
  		 Update changePoint\;
   		}{
   		optimumModel $\gets pl_l$ and $pl_r$\;
   		Update changePoint\;
  	}
  }
 }
 \caption{Multiscale dynamic causal network inference using recursive partitioning.}
 \label{algorithm.RP}
\end{algorithm}

Calculation of the estimator (\ref{rp}) can be accomplished as detailed in Algorithm~\ref{algorithm.RP}.   The required inputs are the time series $\X(u)$ for node $u$, the lagged time series $\X(-u)$ for all other nodes, and a prespecified number of lags $p$. Note that  $p+1$ is the minimum number of observations necessary to fit a model of $p$ lags. Initially we set the penalized likelihood to be the sum of squares of the data in the intervals $I$ that contain less than the minimum required number of observations. There are $(T-1)!$ possible ways of partitioning (i.e., complete recursive partitions $\PP^*$) in the library $\mathcal{L}$. Each partition, however, is composed only of subsets of ${T+1}\choose{2}$ unique intervals, given that each interval is defined between two endpoints. The algorithm begins by fitting group lasso penalized models on intervals $I$ that contain more than $p+1$ observations. Therefore we have $\mathcal{O}(T^2)$ calls for fitting the group lasso type of models. (Because solving the group lasso regression generally requires iterative convex optimization, we do not quantify specifically the corresponding time complexity of this step.) We then consider intervals that contains $2(p+1)$ observations and compare the penalized likelihood $pl_I$ in those intervals to the sum of the penalized likelihoods of the optimal sub intervals containing $p+1$ observations and retain the one with smaller value. The procedure is repeated for intervals containing $k $ observations, with $k = 2(p+1)+1, \cdots, T$. There are $(k-1)$ ways of partitioning an interval of length $k$ into two. Let $\{I_l^i, I_r^i\}_{i=1}^{k=1}$ be all possible pairs of subintervals of $I$ such that $I_l^i \bigcup I_r^i = I$. We compare the penalized likelihood $pl_I$, defined in (\ref{rp}) but restricted to $I$, versus $\min_i\{ pl_{I_l^i} + pl_{I_r^i} + \text{Penalty}\}$, and select the optimal model to be the one which has smallest value. The comparison is of order $\mathcal{O}(T^3)$ and thus the total computational cost is $\mathcal{O}(T^2)$ calls to group lasso type of fitting and $\mathcal{O}(T^3)$ comparisons.

\section{Theoretical properties}
\noindent
In the previous section, we introduced our partition-based approach to modeling dynamical changes in the dependency relational structure among multiple time series, defined two estimators of the time-varying parameters underlying our models, and described an appropriate algorithm for calculations.
 In this section, we first show that the proposed approach can estimate a change point consistently.  We then present an empirically-based choice of the penalty parameter $\lambda$ in equation (\ref{penaltyFunction}) and show that through this choice we can control the Type I error rate in recovering the true neighborhood structure of a node $u$ within a given stationary time block.  Finally, we quantify the overall risk behavior of our estimators.

\subsection{Consistency of changepoint estimation}
\noindent
Suppose that there is a single change point at time $\tau$, with $1 < \tau < T$. Then under our approach the time series $\X(u)$ can be written as a concatenation of two parts of length $\tau$ and $T-\tau$.  We use $L$ to denote the set of all observations in the pre-$\tau$ period and use $R$ to denote the set of all observations in the post-$\tau$ period. Then we have:
\begin{align*}
X_t(u) = 
\left\{
\begin{array}{l}
\sum\limits_{v\in V\backslash \{u\}}\sum\limits_{\ell = 1}^pX_{t-\ell}(v)\theta_L^{(\ell)}(u,v) + \epsilon_t(u),\quad t \in [1, \tau]\\
\sum\limits_{v\in V\backslash \{u\}}\sum\limits_{\ell = 1}^pX_{t-\ell}(v)\theta_R^{(\ell)}(u,v) + \epsilon_t(u),\quad t \in (\tau, T] \enskip .
\end{array}
\right.
\end{align*}
Our change point selection consistency result extends the result of  \cite{bach2008consistency}, where the estimation consistency of the group lasso regression is established. The assumptions needed are the same as in that  previous work, which we briefly restate here.
\begin{Assumption}
$X_t(u)$ and $\X_t(-u)$ have finite fourth order moments: $\mathbb{E}(X_t(u))^4 < \infty$, and $\mathbb{E}\|\X_t(-u)\|^4 < \infty$.
\label{A1}
\end{Assumption}

\begin{Assumption}
Invertibility of the joint covariance matrix, defined as: \\
$
\Sigma_{\X_t(-u)\X_t(-u)} := \E (\X_t(-u)^\prime \X_t(-u)) - \left(\E \X_t(-u)\right)^\prime\left(\E \X_t(-u)\right) \in \mathbb{R}^{(N-1)p \times (N-1)p}
$
\label{A2}
\end{Assumption}\par
\begin{Assumption}
We denote $\boldsymbol{\hat{\theta}_t}$ any minimizer of $\E\left(X_t(u)-\X_t(-u)\boldsymbol\theta_t \right)^2$. We assume that $\E\left(\left(X_t(u)-\X_t(-u)\boldsymbol{\hat\theta}_t \right)^2|\X_t(-u)\right)$ is almost surely greater than some $\sigma_{\min}^2 > 0$.
\label{A3}
\end{Assumption}\par
\begin{Assumption}
$\max\limits_{v\in S^c} \frac{1}{p} \left\| \Sigma_{\X(v)\X(S)} \Sigma_{\X(S)\X(S)}^{-1} \text{Diag}(1/\|\boldsymbol{\theta}_t(u,v)\|_2)\boldsymbol{\theta}_t(u,S)\right\|_2 < 1$,
where $S$ is the set of nodes in the neighborhood of the $u$ where  $(\|\boldsymbol\theta_t(u, v)\|_2 \neq 0)$ and Diag$(1/\|\boldsymbol{\theta}_t(u,v)\|_2)$ denotes the block-diagonal matrix of size $|S|p$ in which each diagonal block equals to $\frac{1}{\|\boldsymbol\theta_t(u,v)\|_2}\mathbf{I}_{|S|p}$ with $\mathbf{I}_{|S|p}$ the identity matrix of size $|S|p$. $\boldsymbol{\theta}_t(u,S)$ denotes the concatenation of the coefficient vectors indexed by $S$.
\label{A4}
\end{Assumption}\par

Note that when $p=1$, Assumption \ref{A4} is referred to as the strong irrepresentable condition in \cite{zhao2006model}.
\begin{Assumption}
The size of the network increases no faster than the square root of the length of the time series: $\exists\, \gamma > 0$, such that $N= \mathcal{O}(T^\gamma)$ as $T \rightarrow \infty$ for $\gamma < 1/2$.
\label{A5}
\end{Assumption}

Consider the local test of
\begin{align*}
H_0 : \PP = [1, T]\quad vs \quad H_1: \PP = [1, \tau] \cup (\tau, T],
\end{align*}
using group lasso penalized least squares.  This test corresponds to the basic step of comparing models for two adjacent intervals at the heart of Algorithm~1 (i.e., one model for the union versus a separate model for each interval), where the penalty is simply the second component of $\text{Pen}_{RP}$ in (\ref{penaltyFunction}).  We have the following theorem:
\begin{Thm}
Assume that Assumptions \ref{A1} to \ref{A5} are satisfied, where $\lambda$ varies such that $\lambda \rightarrow 0$, $\lambda N \rightarrow 0$ and $\lambda T^{1/2} \rightarrow \infty$, as $T\rightarrow\infty$.  Then we have that
\begin{align}
&\mathbb{P}_{H_0}\left(\text{Decide } \PP = [1, T] \right) \longrightarrow 1
\label{equa1}\\
&\mathbb{P}_{H_1}\left(|\hat{\tau} - \tau| > \epsilon \right) \longrightarrow 0,\quad \forall \epsilon > 0 \enskip .
\label{equa2}
\end{align}
\label{splitting}
\end{Thm}\par
Theorem (\ref{splitting}) contains two parts. The first part states that when the null hypothesis is true --  that is, when the time series contains no change point -- our method favors the model with no change point. The second part states that under the alternative hypothesis, where there is a change point at $\tau$, our method favors the model with one estimated change point $\hat\tau$ and, furthermore, the probability that $\hat\tau$ differs from $\tau$ by an arbitrary amount $\epsilon$ tends to zero. The proof can be found in the appendix.   The proof technique can be generalized for the case of multiple change points, although it would require appropriate conditions on the number of change points $M$ and the number of data points $T$.
%Our algorithm is essentially running in a bottom-up fashion and if we have infinite number of observations within each stationary interval, the above result could be generalized to multiple change points detection problem.

\subsection{Finite sample control of Type I error rate in neighborhood selection}
\noindent 
We see that consistent splitting and change point estimation is possible to achieve with the group lasso type of estimation. However, our asymptotic result offers little advice on how to choose a specific penalty parameter for a given problem. We propose a way to adaptively choose the penalty parameters $\lambda$, given a stationary time interval. For a specific $\lambda$, we guarantee that the probability of committing a certain notion of Type I error in recovering the connected component corresponding to the fixed node $u$ is less than some user specified level $\alpha$. The connected component $C_u \in G$ of a node $u \in V$ is defined as the set of nodes which are connected to node $u$ by a chain of edges. We denote the neighborhood of node $u$ as $ne_u$. The neighborhood $ne_u$ is clearly part of the connected component $C_u$. To guarantee the accuracy of the neighborhood selection, we need the following additional assumption:
\begin{Assumption}
Denote by  $\Theta = BV(C)$ the ball of functions of bounded variation for some constant $C$.  We assume that is $\theta_{(\cdot)}^{(\ell)}(u,v)\in \Theta$, for all $\ell = 1,\cdots,p$ and all $v \in V\backslash \{u\}$:
\begin{align*}
\sup_{J \geq 2}\sup_{t_1 \leq \cdots \leq t_J}\sum_{j = p}^J\left|\theta_{t_j}^{(\ell)}(u,\cdot) - \theta_{t_{j-1}}^{(\ell)}(u,\cdot) \right| < C
\end{align*} 
This assumption indicates that $\|\boldsymbol\theta_t(u,v)\|_2$ is bounded.
\label{A6}
\end{Assumption}
In the case where $X(u)$ is stationary on a given interval $[1, T]$, we have the following theorem regarding the estimated connected component $\hat{C_u}$:
\begin{Thm}
Assume Assumptions \ref{A1} to \ref{A6} hold, and fix $\alpha\in (0,1)$.  If $\X(u)$ is stationary on $[1, T]$ and the penalty parameter $\lambda(\alpha)$ is chosen such that
\begin{align*}
\lambda(\alpha) = 2\hat\sigma(u)\sqrt{pQ\left(1 - \frac{\alpha}{N(N-1)}\right)},
\end{align*}
where $\hat\sigma^2(u) = \|\X(u)\|_2^2/T$ and $Q(\cdot)$ is the quantile function of $\chi^2(p)$ distribution, then
\begin{align*}
\mathbb{P}\left(\exists u\in V: \hat{C_u} \nsubseteq C_u \right) \leq \alpha \enskip .
\end{align*}
\label{finite}
\end{Thm}\par
Theorem (\ref{finite}) says that by choosing the penalty parameter at $\lambda = \lambda(\alpha)$, the probability of falsely joining two distinct connected components with the estimate of the edge set is bounded above by the level of $\alpha$. The proof of the theorem is provided in the appendix.

\subsection{Risk analysis}
\noindent
We now provide a theorem that gives an upper bound on the risk of the estimators $\boldsymbol{\hat{\theta}}_{RDP}$ and $\boldsymbol{\hat{\theta}}_{RP}$.  Through this approach we provide a certain measure of quality for the overall dynamic network inference procedure.  Following the perspective of~\cite{li2000mixture}, as implemented in~\cite{kolaczyk2005multiscale}, we measure the loss of estimating $\boldsymbol\theta$ by $\hat{\boldsymbol\theta}$ in terms of the squared Hellinger distance between the two corresponding conditional densities:
\begin{align*}
L(\hat{\boldsymbol\theta}, \boldsymbol\theta) &\equiv H^2(p_{\hat{\boldsymbol\theta}},p_{\boldsymbol\theta}) \\
& = \int \left[\sqrt{p_{\hat{\boldsymbol\theta}}(\x|\X(-u))} - \sqrt{p_{\boldsymbol\theta}(\x|\X(-u))}\right]^2d\nu(\x) 
\end{align*}
with respect to some dominating measure $\nu(\x)$. Additionally, define the Kullback-Leibler divergence between two densities of $\X(u)$, conditional on the past of all the neighborhood time series:
\begin{align*}
K(p_{\boldsymbol\theta^1}, p_{\boldsymbol\theta^2}) \equiv \int \log \frac{p(\x|\X(-u),\boldsymbol\theta^1)}{p(\x|\X(-u),\boldsymbol\theta^2)}p(\x|\X(-u),\boldsymbol\theta^1)d\nu(\x).
\end{align*}

\begin{Thm}
Denote the loss function of estimating $\boldsymbol{\theta}$ by $\boldsymbol{\hat\theta}$ by $L(\boldsymbol{\hat\theta}, \boldsymbol{\theta})$ and the corresponding risk, by $R(\boldsymbol{\hat\theta}, \boldsymbol{\theta}) = T^{-1}\mathbb{E}_{\X(u)|\X(-u)}
\left[ L(\boldsymbol{\hat\theta}, \boldsymbol{\theta}) \right]$. Let $\Lambda = {\alpha_{max}/T}$, where $\alpha_{max}$ is the largest eigenvalue of $\X(-u)^\prime \X(-u)$. Assume each $\theta^{(\ell)}_t (u,v)$ is of bounded variation on $(0, 1]$ for some constant $C$. Then for any $\lambda$ of the same order as in Theorem~\ref{splitting} and for $T > \lceil e^{2p/3} \rceil$, our risk is bounded as
$$R(\boldsymbol{\hat\theta}_{RDP}, \boldsymbol{\theta}) \le \mathcal{O}\left(\left(
\frac{\Lambda \log^4T}{T}\right)^{1/3}\right) \enskip $$
for recursive dyadic partitioning and
$$R(\boldsymbol{\hat\theta}_{RP}, \boldsymbol{\theta}) \le \mathcal{O}\left(\left(
\frac{\Lambda \log^2T}{T}\right)^{1/3}\right) \enskip $$ for recursive partitioning.
\label{Riskbound}
\end{Thm}\par
Theorem \ref{Riskbound} shows that both estimators have risks that end to zero at rates slightly worse than $T^{-1/3}$. The asymptotic risk for recursive partitioning is smaller than the risk for recursive dyadic partitioning, albeit at the cost of increased computational complexity.  Proof of this result is in line with the work by \cite{kolaczyk2005multiscale} and can be found in the appendix.

\section{Simulation study}
\noindent
In this section, we illustrate the practical performance of our method through a series of simulation studies. In the first part, we simulate multivariate time series data under different settings, as dictated by models A - C below. In the second part, we scale up model B by increasing the size of the vertex set $V$ and include more irrelevant variables. Under each model, we simulate 100 datasets and the white noise is always set to be $\epsilon_t(\cdot) \sim N(0,1)$. In all models, we set $\alpha = 0.05$ and $p = 2$. These choices match that of the computational neuro-science example we present later, in Section~5.  We measure performance in three ways: (i) how many change points were detected, (ii) Out of the detected change points, how many specify the right location  (iii) whether the correct neighborhood structure was detected. The models we investigate are:
\begin{itemize}
\item Model A: VAR(2) process with no change point.\\
This scenario is designed to see the performance of the methods when there is no change point and the process is stationary.  Specifically,
\begin{align*}
X_t(1) = 0.5 X_{t-1}(2) + 0.25X_{t-2}(2) + 0.5 X_{t-1}(3) + 0.25X_{t-2}(3) + \epsilon_t(1)
\end{align*}
with sample size $T = 1024$.
\item Model B: piecewise stationary VAR(2) process with 2 change points.\\
Specifically,
\begin{align*}
X_t(1) = \left\{
\begin{array}{l}
0.5X_{t-1}(2) + 0.25X_{t-2}(2)+\epsilon_t(1) \quad \quad 0 < t \leq 512\\
0.5X_{t-1}(3) + 0.25X_{t-2}(3) + \epsilon_t(1) \quad \quad 512<t \leq 768\\
0.5X_{t-1}(2) - 0.5X_{t-1}(3)+\epsilon_t(1) \,\, \,\quad \quad 768 < t \leq 1024\\
\end{array}\right.
\end{align*}
\item Model C: change point close to the boundary.  \\
Specifically,
\begin{align*}
X_t(1) = \left\{
\begin{array}{l}
0.5X_{t-1}(2) + 0.25X_{t-2}(2) + \epsilon_t(1) \quad \quad 0 < t \leq 128\\
0.5X_{t-1}(3) + 0.25X_{t-2}(3) + \epsilon_t(1) \quad \quad 128<t \leq 1024
\end{array}\right.
\end{align*}
\item Model B with VAR(2) process in a larger vertex set $V$.\\
We use the same coefficients as used in Model B, but with the size of the vertex set ranging from $5$ to $15$.
\end{itemize}

The results for models A, B, and C are summarized in Table~\ref{tab:table1}. For some error measures, results under the truth are marked in blue. For example, under model A where there is no change point in the true model, positions corresponding to 0 change point and 0 exact detection are marked in blue, i.e., one should not detect anything where there is no change point. Under model B, where there are two change points, results corresponding to the case of two change points and two exact detections are marked in blue.  Note that in the case recursive partitioning (i.e., non-dyadic), we treat a detection as being 'exact' if an estimated change point is within $\pm 5$ time points of the true change point (i.e., less than $0.5 \%$ the length of the full time series).

\begin{table}
\centering
\resizebox{0.9\textwidth}{!}{\begin{minipage}{1.1\textwidth}
\begin{tabular}{cp{1.5cm}p{1.5cm}p{1.5cm}p{1.5cm}p{1.5cm}p{1.5cm}}
	\toprule
	 & \multicolumn{3}{c}{\textbf{RDP}} &\multicolumn{3}{c}{\textbf{RP}}\\
	\cmidrule(r){2-4} 
	\cmidrule(r){5-7} 
	\textbf{Model} & Model A & Model B& Model C& Model A & Model B & Model C\\
	\midrule
	\textbf{\# change point} & & & & & & \\
	$0$& \textcolor{blue}{100} & 0 & 100 & \textcolor{blue}{100} & 0 & 89\\
	$1$& 0 & 28 & \textcolor{blue}{0} & 0 & 0 & \textcolor{blue}{11}\\
	$ 2$& 0 & \textcolor{blue}{72} & 0 & 0 & \textcolor{blue}{100}& 0\\
	\midrule
	\textbf{\# exact detection} & & & & & &\\
	$0$& \textcolor{blue}{100} & 0 & 100 &\textcolor{blue}{100}  &0 & 89\\
	$1$& 0 & 28 &  \textcolor{blue}{0}  & 0& 11& \textcolor{blue}{11}\\
	$ 2$& 0& \textcolor{blue}{72}& 0 & 0& \textcolor{blue}{89}&0\\
	\midrule
	\textbf{\# false edge detection} & & & & & &\\
	$0$& {100} & {97} & {100} & {100} &{94} &{100}\\
	$1$& 0 & 3 & 0 & 0& 6 & 0 \\
	$2$& 0 & 0 & 0 & 0& 0& 0\\
	\bottomrule
		\end{tabular}
\caption{Simulation results under Model A, Model B and Model C, using RDP and RP.}
\label{tab:table1} 
\end{minipage} }
\end{table}

A few comments on these results are in order:
\begin{itemize}
\item From the results we see that our proposed estimators did not overestimate the number of change points, as they never detected more change points than the true number of change points.
\item Under Model B, in 72 out of 100 and in 89 out of 100 trials we correctly specified the number of positions of the change points using the recursive dyadic partition and the recursive partition estimators, respectively. Note that if we are less conservative and allow more tolerance in defining an `exact detection' under recursive partitioning, all change points identified in Model B using recursive partitioning are located within $[-13, 13]$ points of the true change points (i.e., within $1.5\%$ of the total length of the full time series).
\item Based on the results under model C, we conclude that our methods lose sensitivity to detection of change points as the location of the change points moves closer to the boundary, with recursive partitioning performing better than recursive dyadic partitioning.  These results are to be expected.
\item We have good control over the false detection of causal structures.
\end{itemize}
\begin{table}
\centering
\resizebox{0.9\textwidth}{!}{\begin{minipage}{1.2\textwidth}
\begin{tabular}{cp{1cm}p{1cm}p{1cm}p{1cm}p{1cm}p{1cm}p{1cm}p{1cm}p{1cm}p{1cm}}
	\toprule
	 & \multicolumn{5}{c}{\textbf{RDP}} &\multicolumn{5}{c}{\textbf{RP}}\\
	\cmidrule(r){2-6} 
	\cmidrule(r){7-11} 
	\textbf{Size of $V$} & 5 & 7& 11& 13 & 15 & 5 & 7& 11& 13 & 15 \\
	\midrule
	\# \textbf{change point} & & & & & & & & & & \\
	$0$&22 & 48 & 67 & 95&100& 0& 19& 52& 93 &100 \\
	$1$& 20 &20 & 29 &5 &0 & 0 & 4& 2& 0& 0\\
	$ 2$&\textcolor{blue}{58} &\textcolor{blue}{32} &\textcolor{blue}{4} &\textcolor{blue}{0} &\textcolor{blue}{0} &\textcolor{blue}{100} &\textcolor{blue}{77} & \textcolor{blue}{46}&\textcolor{blue}{7} & \textcolor{blue}{0}\\
	\midrule
	\textbf{\# exact detection}& 0& & & & & & & & &\\
	$0$&22 & 48 & 67 & 95&100& 0& 19&52 &93 &100 \\
	$1$& 20 &20 & 29 &5 &0 & 17& 19& 4 & 0& 0\\
	$ 2$&\textcolor{blue}{58} &\textcolor{blue}{32} &\textcolor{blue}{4} &\textcolor{blue}{0} &\textcolor{blue}{0} & \textcolor{blue}{83}& \textcolor{blue}{62}& \textcolor{blue}{44}&\textcolor{blue}{7} &\textcolor{blue}{0} \\
	\midrule
	\textbf{\# false edge detection} & & & & & & & & & &\\
	$0$& {98}& {100} & {100}& {100}&{100}& {98}&{97} & {100} & {100}& {100}\\
	$1$& 2 & 0& 0&0 &0 & 2& 3 &0 &0 &0\\
	$2$& 0&0 &0 &0 &0 &0 &0 &0 & 0&0\\
	\bottomrule
	\end{tabular}
\caption{Simulation results under Model B for vertex sets of increasing cardinality.}
\label{tab:table2} 
\end{minipage} }
\end{table}
The performance of the proposed estimators upon increasing the size $N$ of the vertex set $V$, under model B, is summarized in Table~\ref{tab:table2}. As $N$ increases, we see the performance decreases, due to the fact that in this setting the variables we are adding are irrelevant and thus induce additional uncertainty. Note that under our proposed approach there is a tendency to underfit the number of change points rather than over fit.  This trait will be relevant to the real data application we describe next.

\section{Illustration: Inference of a task-based MEG network}
\noindent
Neuroscientists are interested in understanding the interactions among cortical areas that allow subjects to detect the motion of objects. In~\cite{calabro2012interaction}, fMRI was used to study subjects who were asked to perform visual search tasks and it was found that the monitored regions of interest (ROIs) formed four clusters.  However, fMRI does not have good temporal resolution for more detailed investigation of the interaction between these clusters. \cite{rana2014functional} studied the 10 Hz Alpha-band power extracted from MEG signals under a similar multiple-trial visual motion search experiment. They found evidence showing that regions of interest within the identified clusters have similar temporal activation profiles. Specifically, they found significant inhibition of 10Hz alpha power in the visual processing region after 300ms relative to the stimulus, and longer and sustained alpha power in the frontoparietal region. Other evidence of co-activations among regions of interest have been reported by other studies under different experimental set up. For example, see \cite{braddick2000form}, \cite{amano2012human} and \cite{bettencourt2016decoding}.

To demonstrate the application of our method, we examined the same 10 Hz Alpha-band power data used by \cite{rana2014functional}. MEG data has excellent temporal resolution, but the spatial resolution is less good than that of fMRI. As a result, it is typical that functional connectivity analyses with MEG data incorporate coarsely defined brain regions and hence networks with only a handful of vertices. We therefore chose three regions of interest each from the two clusters known to have similar activation profiles. The regions of interest are V3a, MT+ and VIP from the visual processing region, and FEF, SPL and DLPFC from the frontoparietal region. This choice corresponds to a network of six nodes, which is consistent with studies of this type.
%We found that the 10Hz alpha power is synchronized among ROIs within the two clusters and is independent between the two. 

Details of the experiment and the data are as follows. In the experiment, a participant was asked to perform a visual search task of a moving object, repeated over $160$ trials. Each trial began with a 300 ms blank screen. Then, 9 spheres fade in over a 1000 ms period and these 9 spheres remained static for another 1000 ms. A 1000 ms motion display period then follows, where 8 of the spheres move forward (simulating forward motion of the obeserver) and the target sphere moves independently from the others. The beginning of the motion display period defined the 0 ms marker for each trial. Finally, in the 3000 ms response period, the 9 spheres remained static, four (including the target) were grayed out, and the participant was asked to identify the target sphere. 

The MEG signal of the participant was recorded throughout the experiment. The data we used is the 10 Hz Alpha-band power, truncated in a uniform manner across trials, to focus upon the period just prior to the appearance and movement of the spheres. It starts from the second half of the static period and the length of the data is $T = 1502$, corresponding to a time interval of length $2500$ ms. The time series we used for our analyses contains the last $500$ ms of the static period, the entire motion display period, and the first $1000$ ms of the response period, where most of the correct responses occurred. The timeline of our data is illustrated by Fig~\ref{experiment}. For a more detailed description of the experiment, please refer to \cite{rana2014functional}.
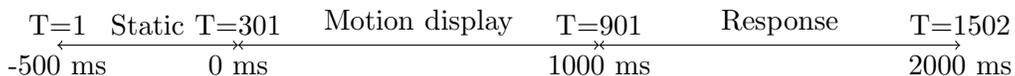
\begin{figure}[!htb]
\centering
\begin{tikzpicture}[scale=0.8]
\draw [<->] (0,0) -- (3,0)
node[pos=0.5,above] {Static}
node[pos=0,above] {T=1}
node[pos=0,below] {-500 ms};	
\draw [<->] (3,0) -- (9,0)
node[pos=0,above] {T=301}
node[pos=0,below] {0 ms}
node[pos=0.5,above] {Motion display};
\draw [<->] (9,0) -- (15,0)
node[pos=0,below] {1000 ms}
node[pos=0,above] {T=901}
node[pos=0.5,above] {Response}
node[pos=1,below] {2000 ms}
node[pos=1,above] {T=1502};
\end{tikzpicture}
\caption{Visual search experiment time line.}
\label{experiment}
\end{figure}

Each time series has been pre-processed by taking the first order difference to remove the self-driven component. We then use the recursive partition based method with lag $p = 7$ (chosen in preliminary analysis using the Akaike information criterion). We set the level $\alpha$ in Theorem~\ref{finite} to $0.05$. The recursive dyadic method does not apply here because the length of the data is not a power of $2$. 
 \begin{figure}[H]
        \centering
        \begin{subfigure}[b]{0.475\textwidth}
            \centering
            \includegraphics[width = 60mm, height = 60mm]{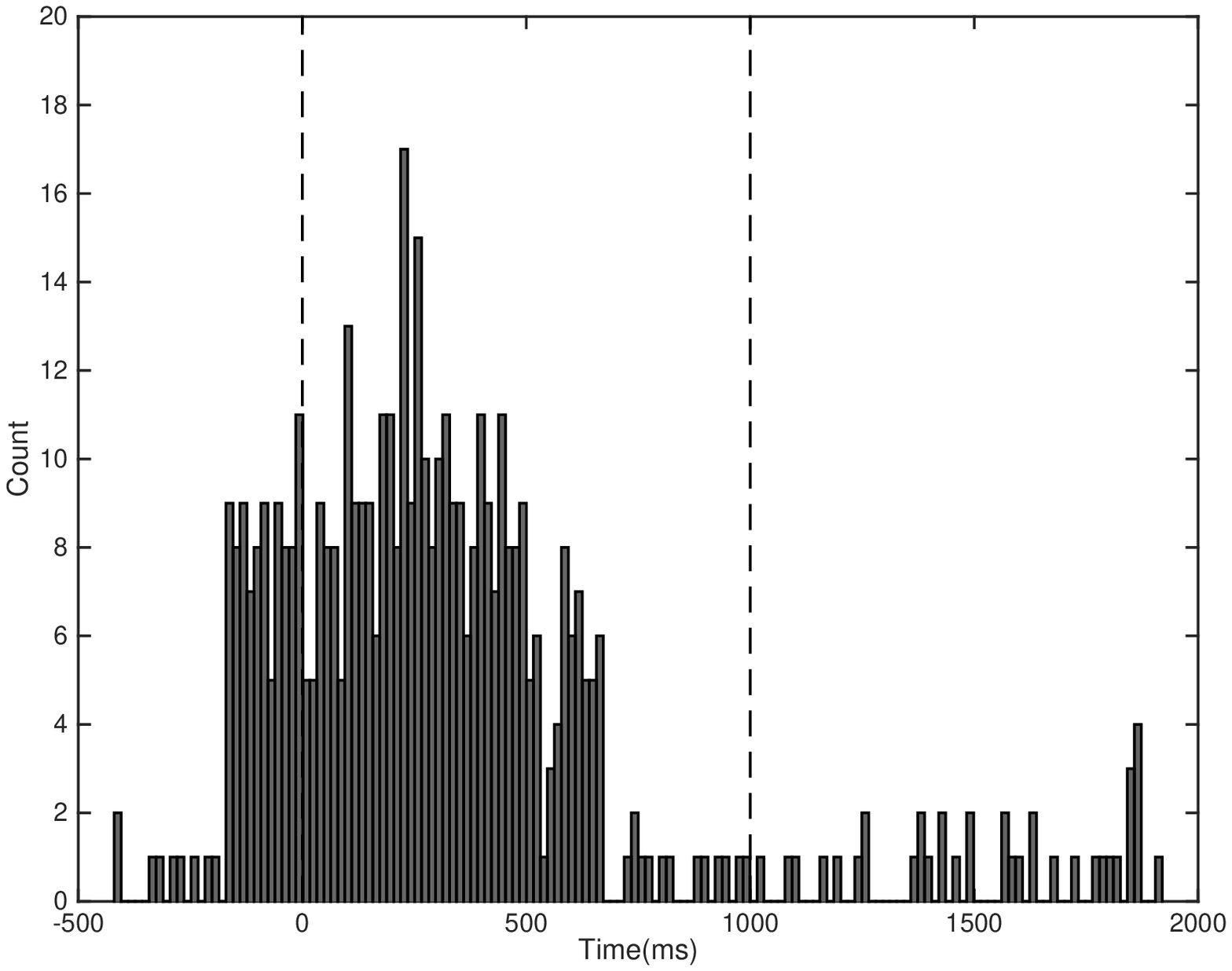}
            \caption[Network2]%
            {{\small Distribution of change points among the visual processing region.}}    
            \label{fig:CP_VP}
        \end{subfigure}
        \hfill
        \begin{subfigure}[b]{0.475\textwidth}  
            \centering 
            \includegraphics[width = 60mm, height = 60mm]{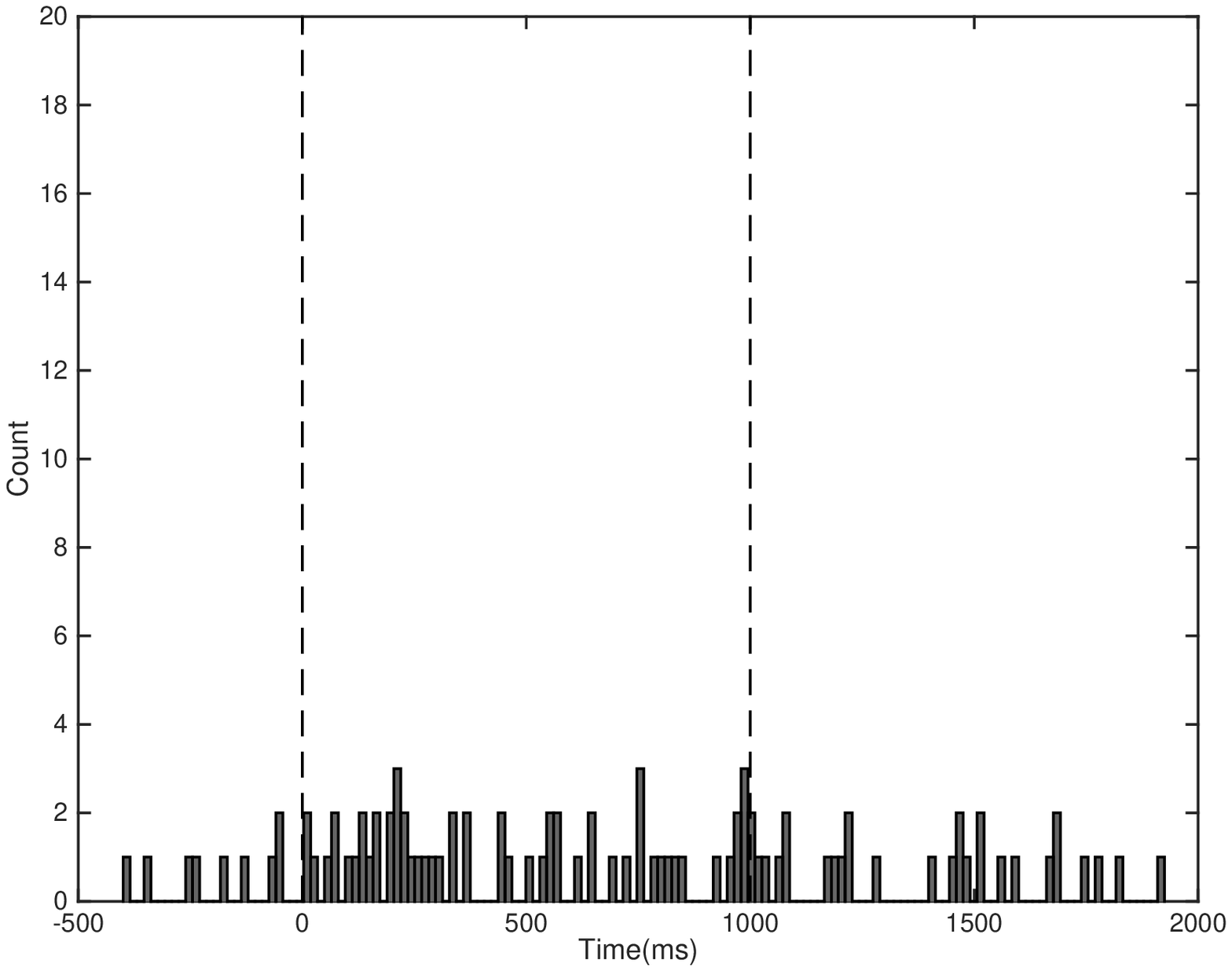}
            \caption[]%
            {{\small Distribution of change points among the frontoparietal region.}}    
            \label{fig:CP_FP}
        \end{subfigure}
        \caption[ ]
        {\small The change point distribution among the visual processing region and the frontoparietal region.} 
        \label{fig:CP}
   \end{figure}
Fig \ref{fig:CP} shows the distribution of the detected change points among each of the two clusters we examined. The two dashed vertical lines indicate the time of the two phase changes. There are $497$ change points detected across the $160$ trials in the visual search region, of which $427$ lie between -150 ms and 750 ms, relative to the stimulus onset. Compared with the visual processing regions, there are much fewer change points detected among the frontoparietal regions, where the Alpha-band power is more sustained.

Strength of the connections between regions of interest, within each of the two clusters, is shown in Figures~\ref{fig:Coef_VP} and~\ref{fig:Coef_FP}, where we have plotted the pointwise means and one standard deviation error bars of the $\ell_2$ norm of the coefficients across the $160$ trials. The inhibitive role of the Alpha-band power in the visual processing region (i..e, the creation of a common co-deactivation pattern), in response to the stimulus, is understood to be the reason for the significant increase in the $\ell_2$ norms of the coefficients among V3a, MT+ and VIP from -150 ms to 750 ms.  And, in fact, most of the changepoints in this time interval among these three regions of interest correspond to an increase in the $\ell_2$ norm of the pair-wise regression coefficients.  In contrast, the changes of the $\ell_2$ norms of the coefficients in the frontoparietal region are much more gradual.
\begin{center}
 \begin{figure}[H]
        \centering
        \begin{subfigure}[b]{0.3\textwidth}
            \centering
            \includegraphics[width = 50mm, height = 50mm]{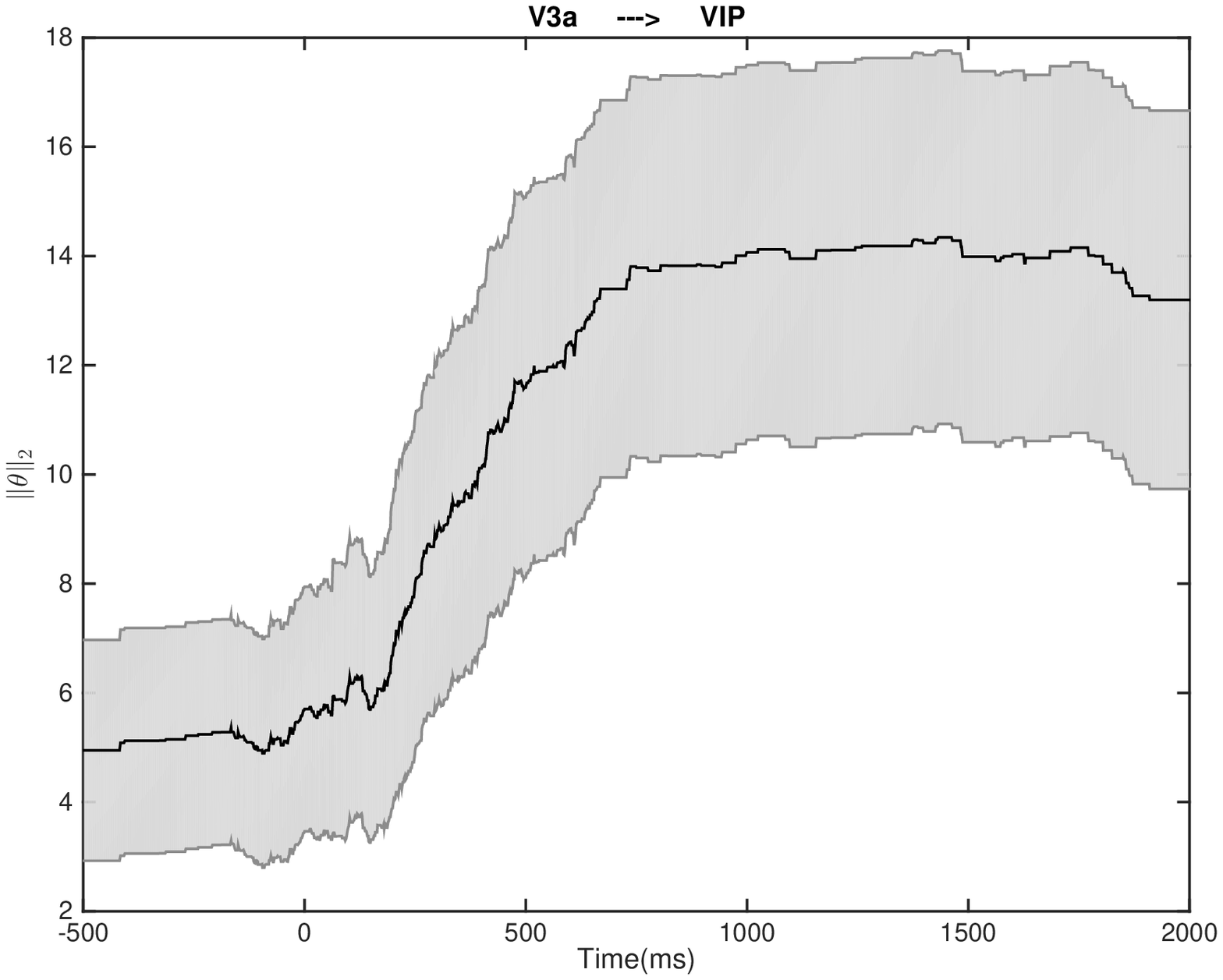}
            \caption[Network2]%
            {{\small $\ell_2$ norm of edge V3a $\rightarrow$ VIP.}}    
            \label{fig:v3a-vip}
        \end{subfigure}
        \hfill
        \begin{subfigure}[b]{0.3\textwidth}  
            \centering 
            \includegraphics[width = 50mm, height = 50mm]{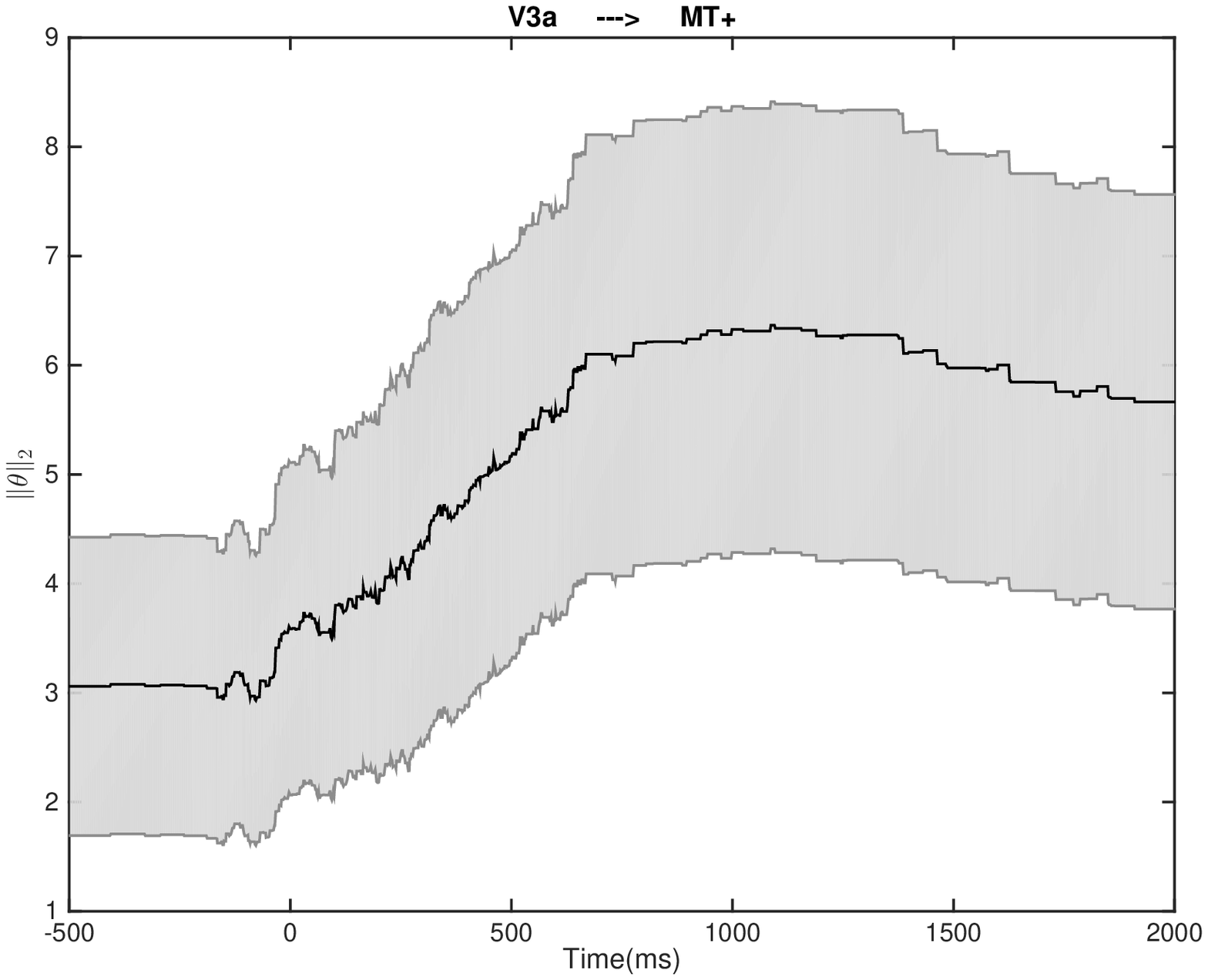}
            \caption[]%
            {{\small$\ell_2$ norm of edge V3a $\rightarrow$ MT+.}}    
            \label{fig:v3a-mtp}
        \end{subfigure}
        \hfill
            \begin{subfigure}[b]{0.3\textwidth}  
            \centering 
            \includegraphics[width = 50mm, height = 50mm]{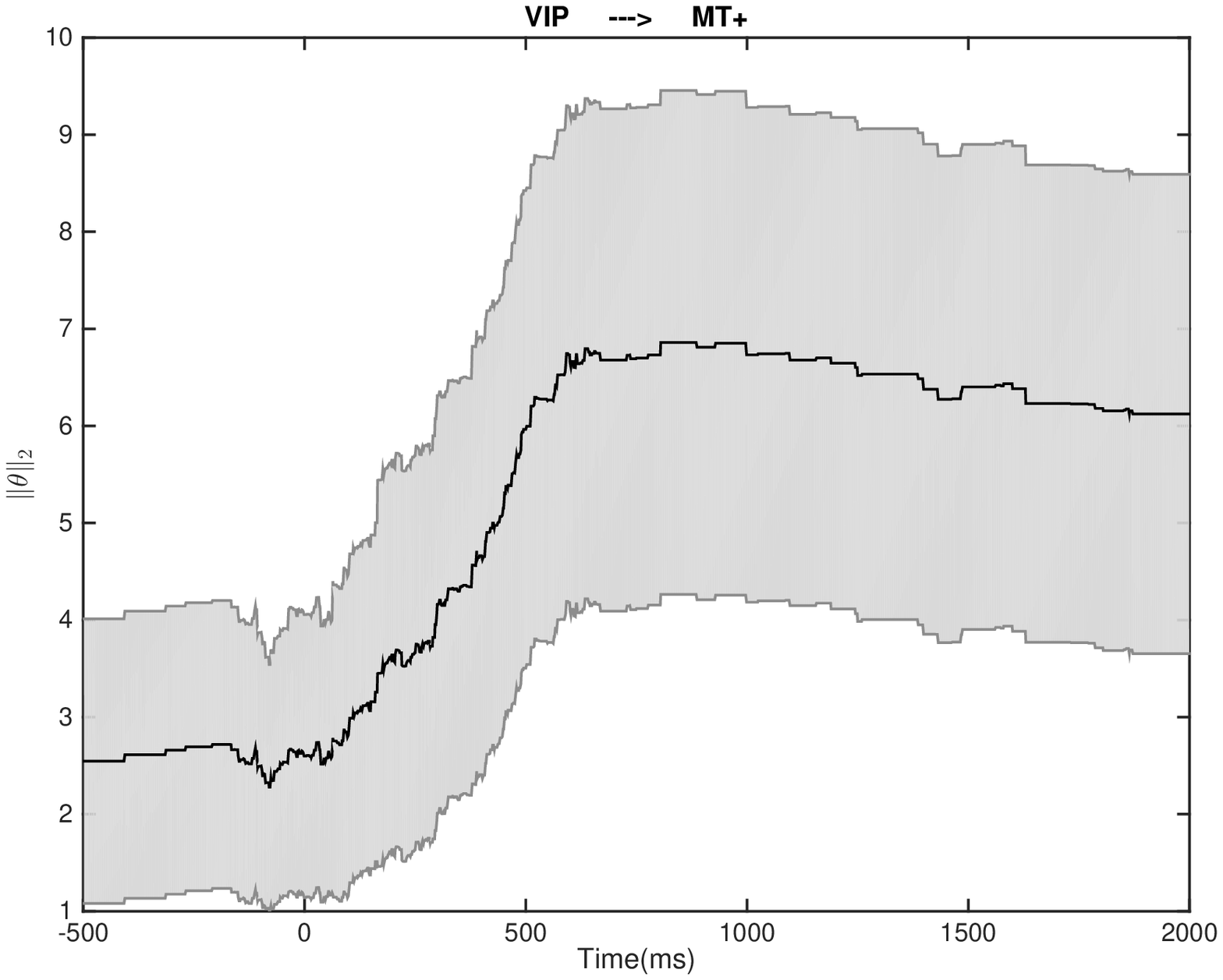}
            \caption[]%
            {{\small $\ell_2$ norm of edge VIP $\rightarrow$ MT+.}}    
            \label{fig:vip-mtp}
        \end{subfigure}
        \newline
                \begin{subfigure}[b]{0.3\textwidth}
            \centering
            \includegraphics[width = 50mm, height = 50mm]{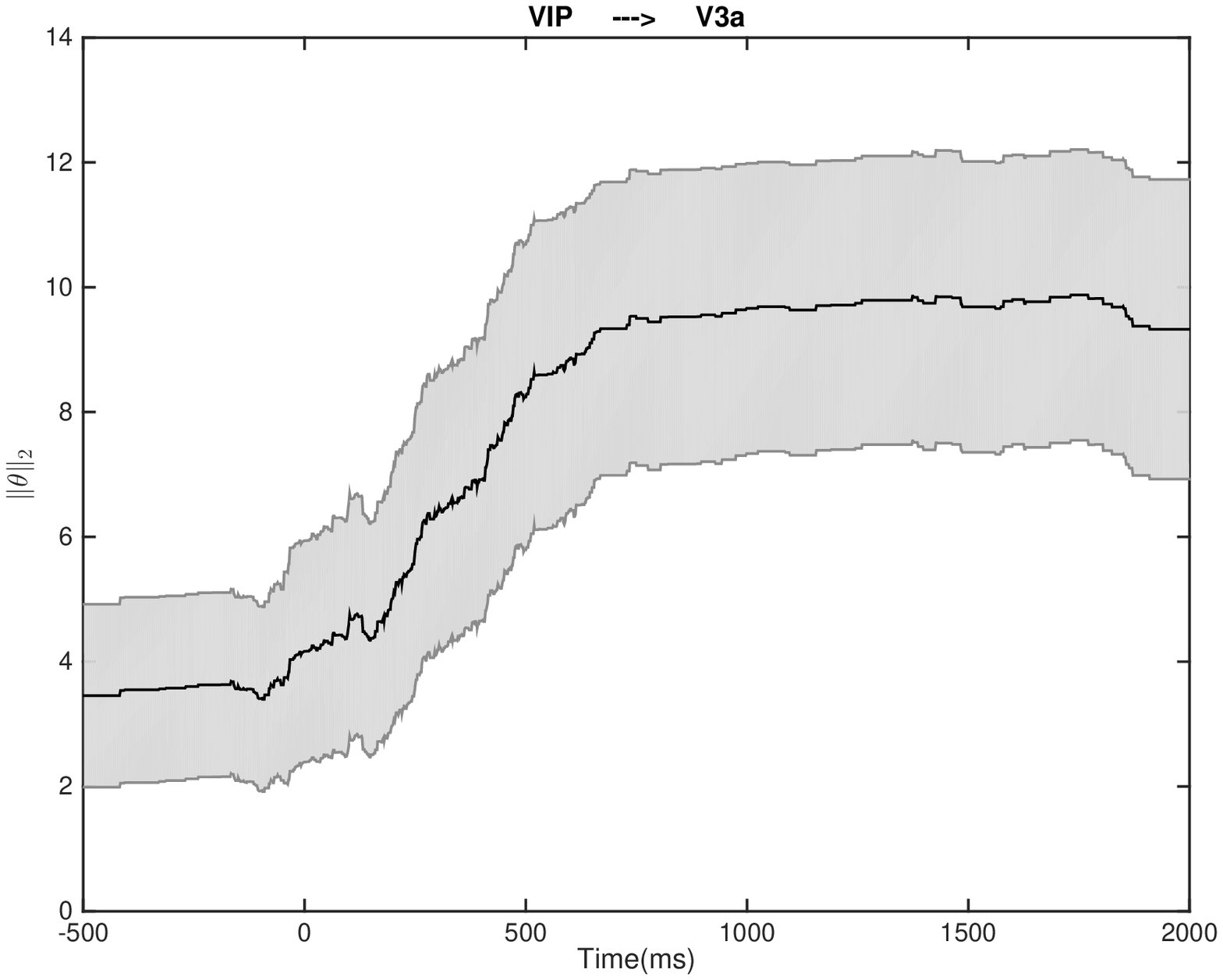}
            \caption[Network2]%
            {{\small $\ell_2$ norm of edge VIP $\rightarrow$ V3a.}}    
            \label{fig:vip-v3a}
        \end{subfigure}
        \hfill
        \begin{subfigure}[b]{0.3\textwidth}  
            \centering 
            \includegraphics[width = 50mm, height = 50mm]{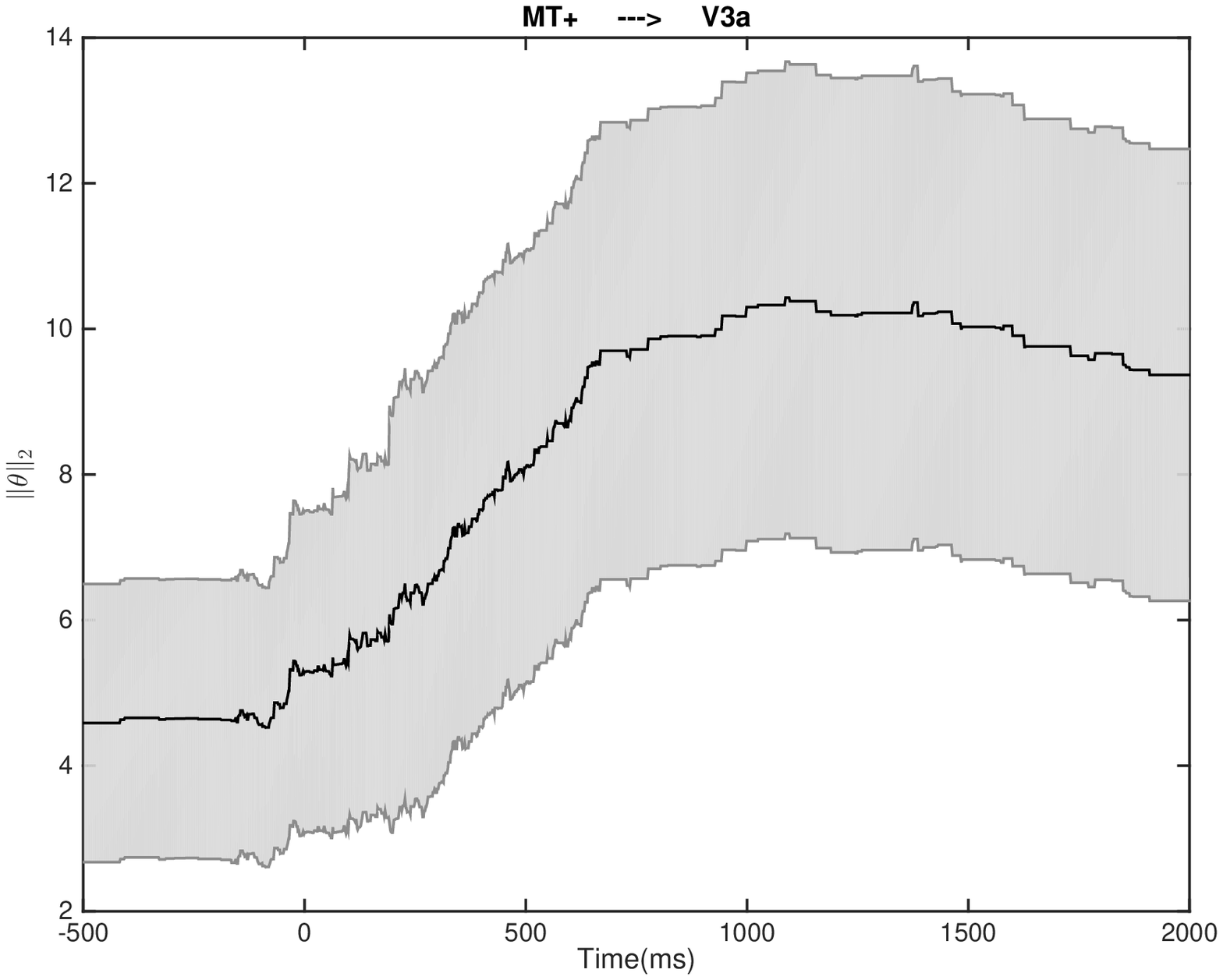}
            \caption[]%
            {{\small $\ell_2$ norm of edge MT+ $\rightarrow$ V3a.}}    
            \label{fig:mtp-v3a}
        \end{subfigure}
        \hfill
            \begin{subfigure}[b]{0.3\textwidth}  
            \centering 
            \includegraphics[width = 50mm, height = 50mm]{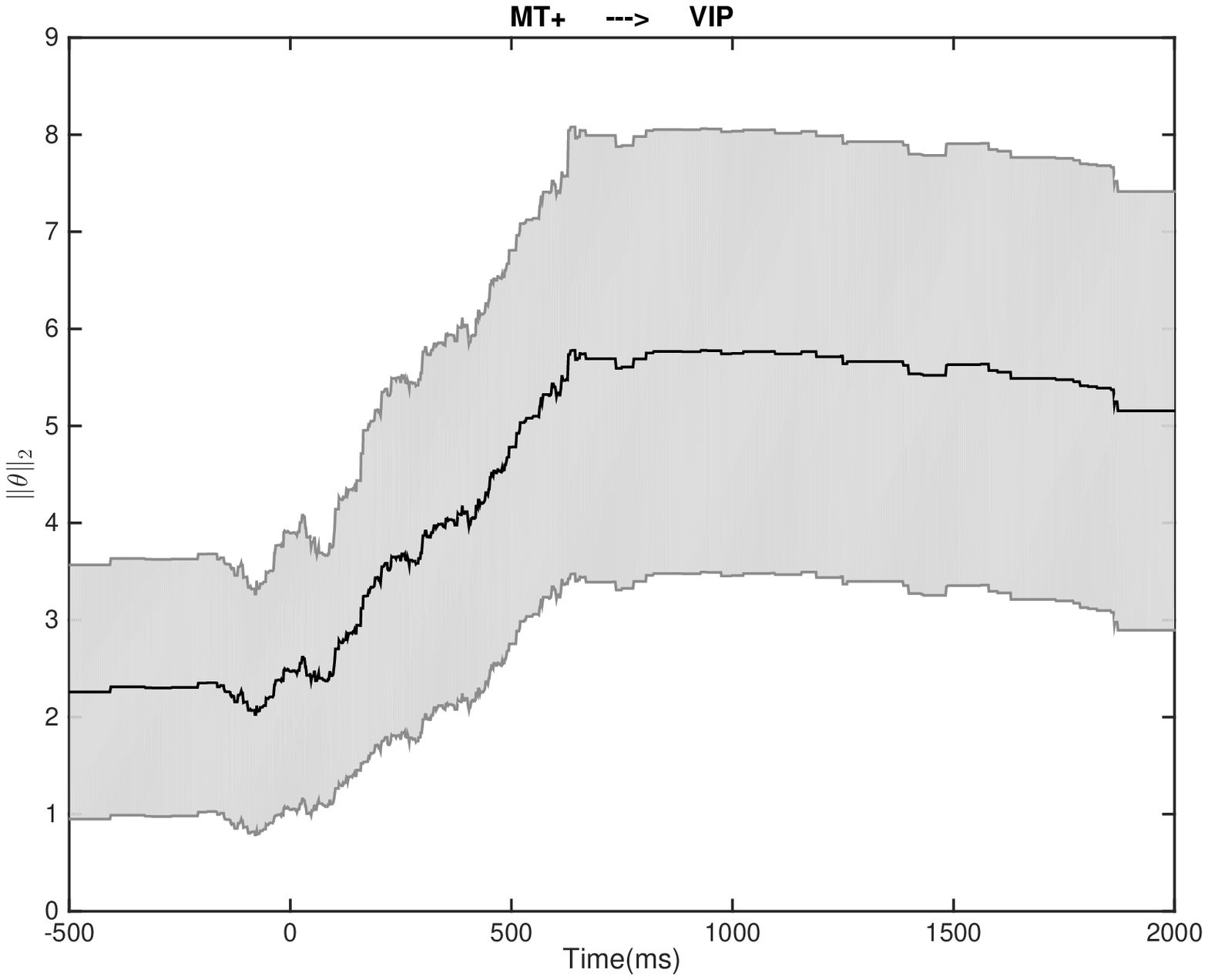}
            \caption[]%
            {{\small $\ell_2$ norm of edge MT+ $\rightarrow$ VIP.}}    
            \label{fig:mtp-vip}
        \end{subfigure}
        \caption[ ]
        {\small $\ell_2$ norms of coefficients between pairs of time series in the visual processing region.} 
        \label{fig:Coef_VP}
    \end{figure}
\end{center}
\begin{center}
 \begin{figure}[H]
        \centering
        \begin{subfigure}[b]{0.3\textwidth}
            \centering
            \includegraphics[width = 50mm, height = 50mm]{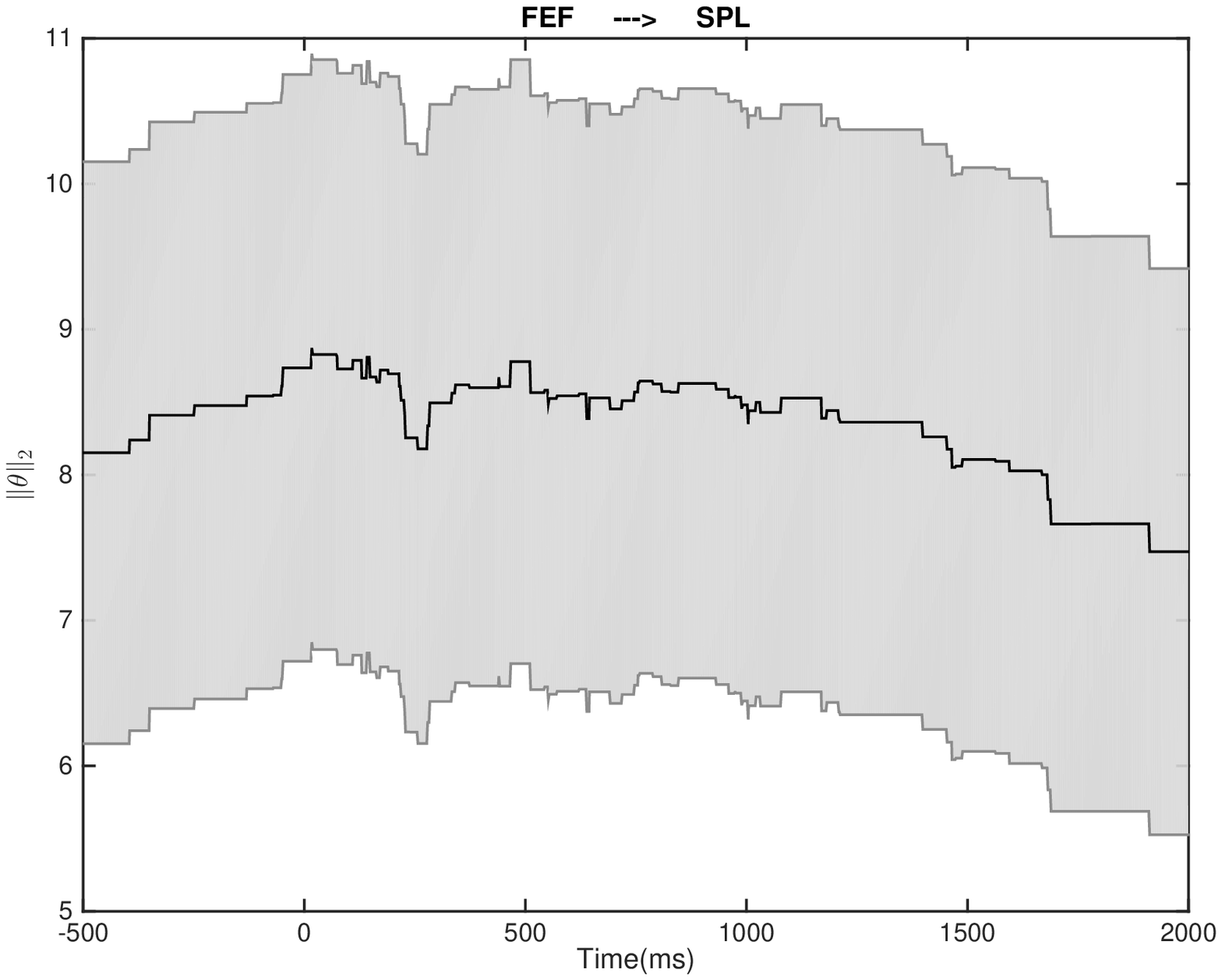}
            \caption[Network2]%
            {{\small $\ell_2$ norm of edge FEF $\rightarrow$ SPL.}}    
            \label{fig:fef-spl}
        \end{subfigure}
        \hfill
        \begin{subfigure}[b]{0.3\textwidth}  
            \centering 
            \includegraphics[width = 50mm, height = 50mm]{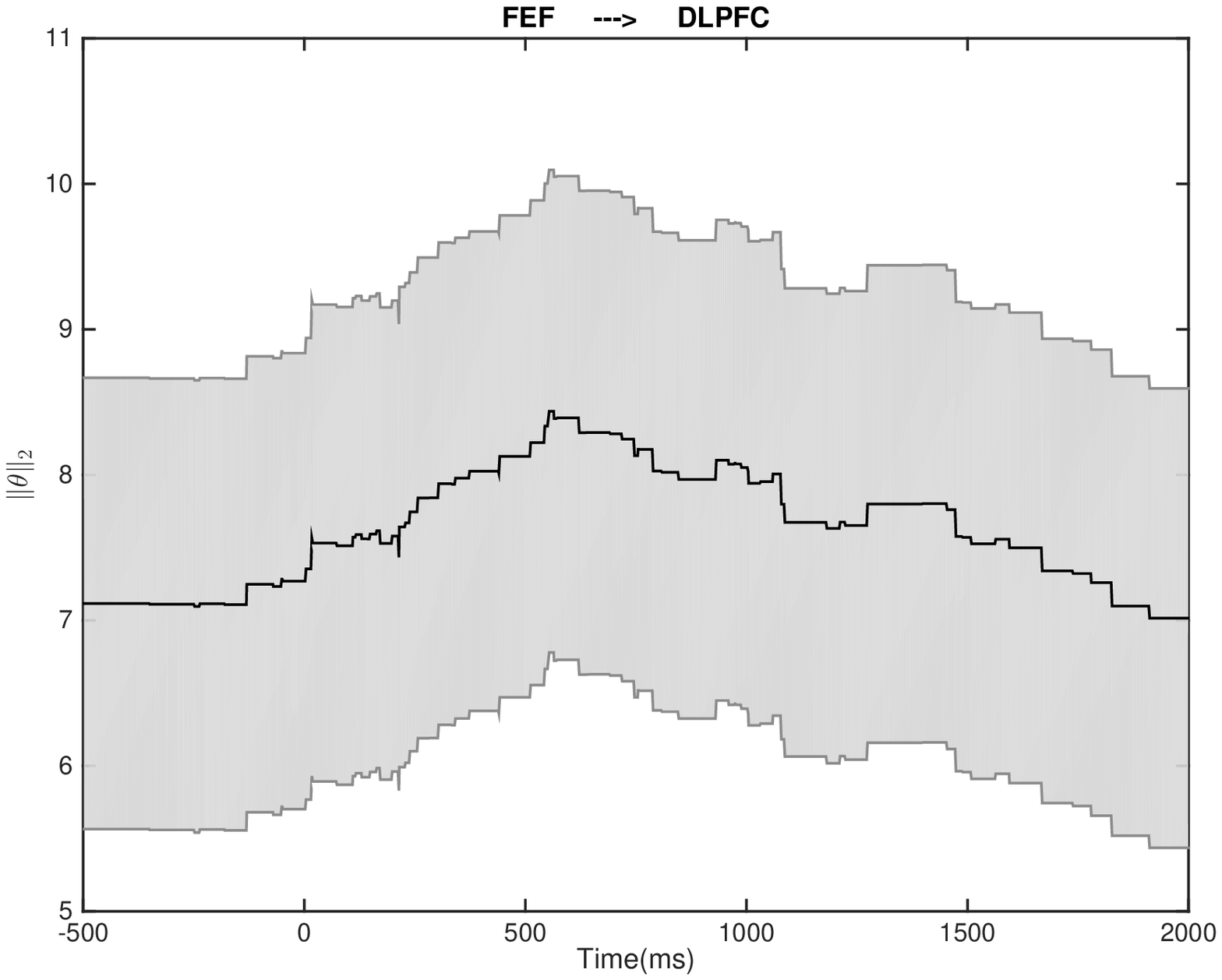}
            \caption[]%
            {{\small $\ell_2$ norm of edge FEF $\rightarrow$ DLPFC.}}    
            \label{fig:fef-dlpfc}
        \end{subfigure}
        \hfill
            \begin{subfigure}[b]{0.3\textwidth}  
            \centering 
            \includegraphics[width = 50mm, height = 50mm]{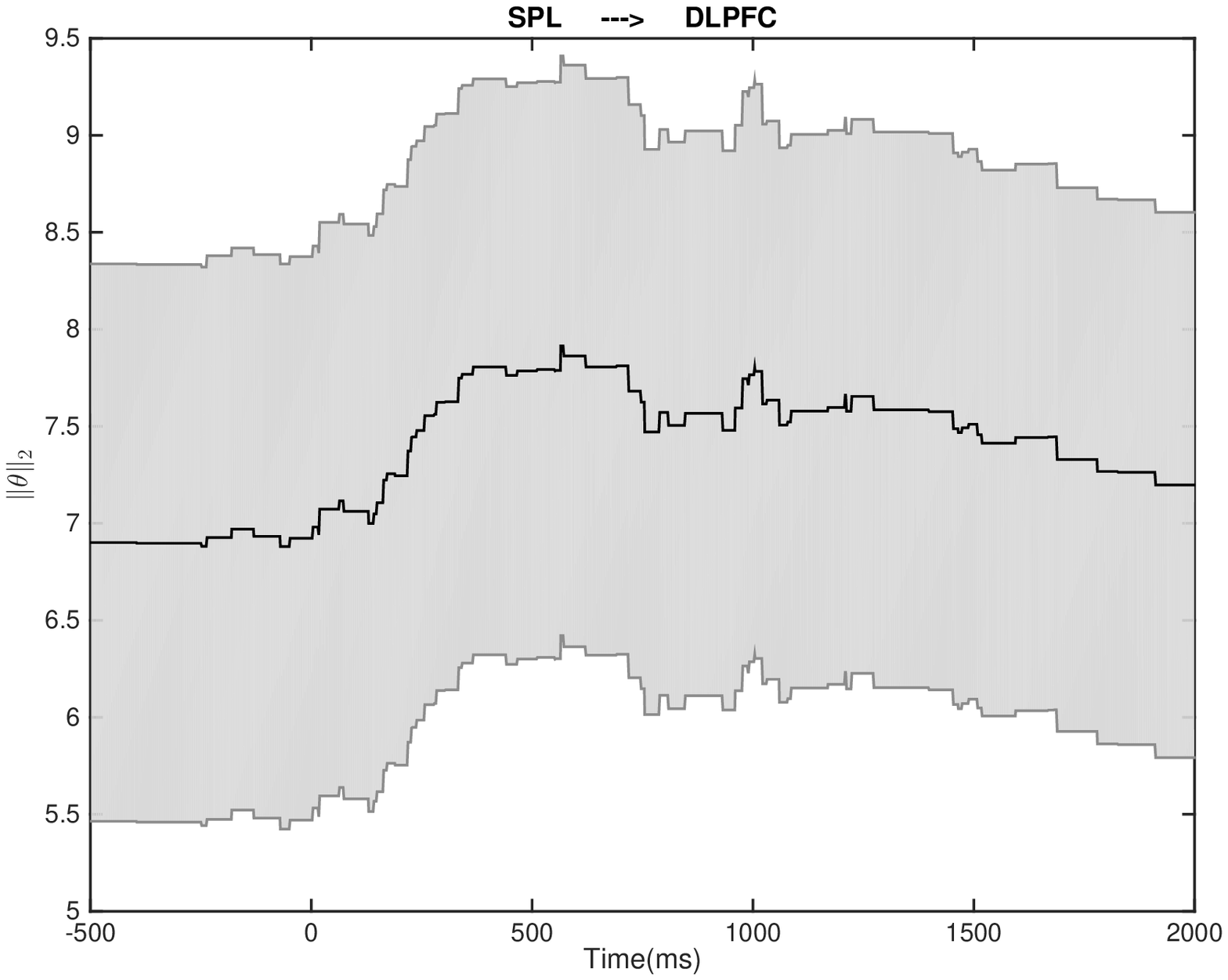}
            \caption[]%
            {{\small $\ell_2$ norm of edge SPL $\rightarrow$ DLPFC.}}    
            \label{fig:spl-dlpfc}
        \end{subfigure}
        \newline
                \begin{subfigure}[b]{0.3\textwidth}
            \centering
            \includegraphics[width = 50mm, height = 50mm]{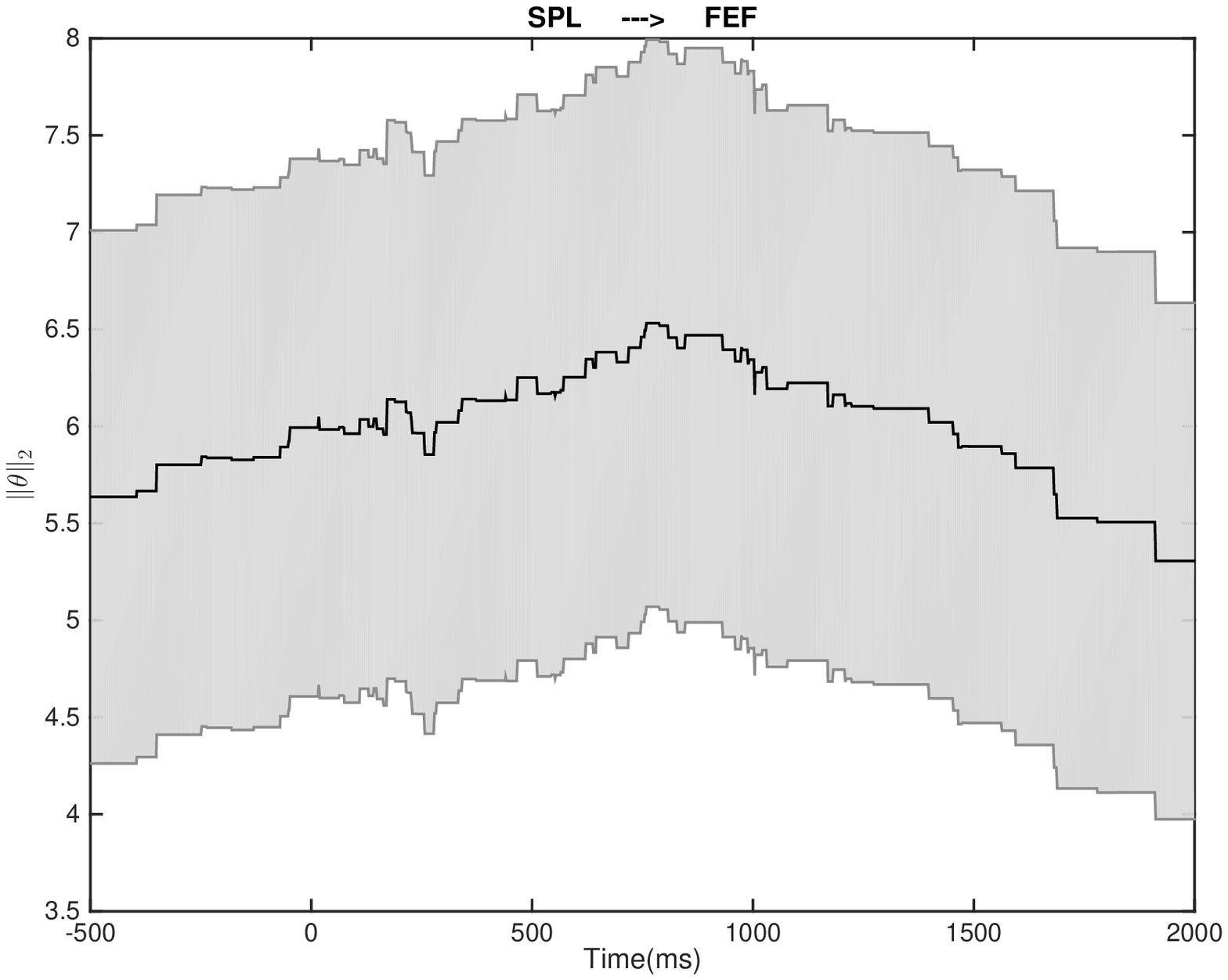}
            \caption[Network2]%
            {{\small $\ell_2$ norm of edge SPL $\rightarrow$ FEF.}}    
            \label{fig:spl-fef}
        \end{subfigure}
        \hfill
        \begin{subfigure}[b]{0.3\textwidth}  
            \centering 
            \includegraphics[width = 50mm, height = 50mm]{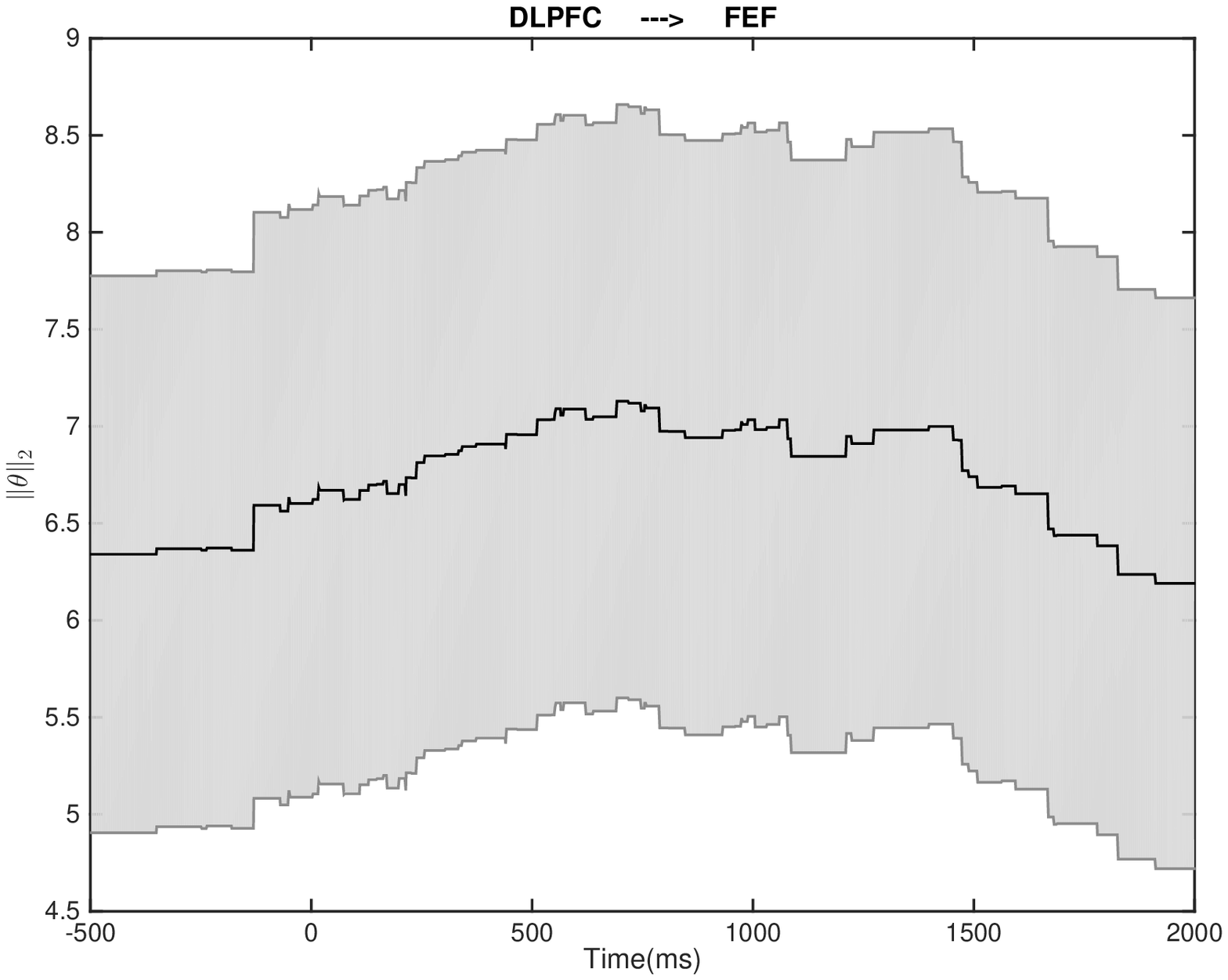}
            \caption[]%
            {{\small $\ell_2$ norm of edge DLPFC $\rightarrow$ FEF.}}    
            \label{fig:dlpfc-fef}
        \end{subfigure}
        \hfill
            \begin{subfigure}[b]{0.3\textwidth}  
            \centering 
            \includegraphics[width = 50mm, height = 50mm]{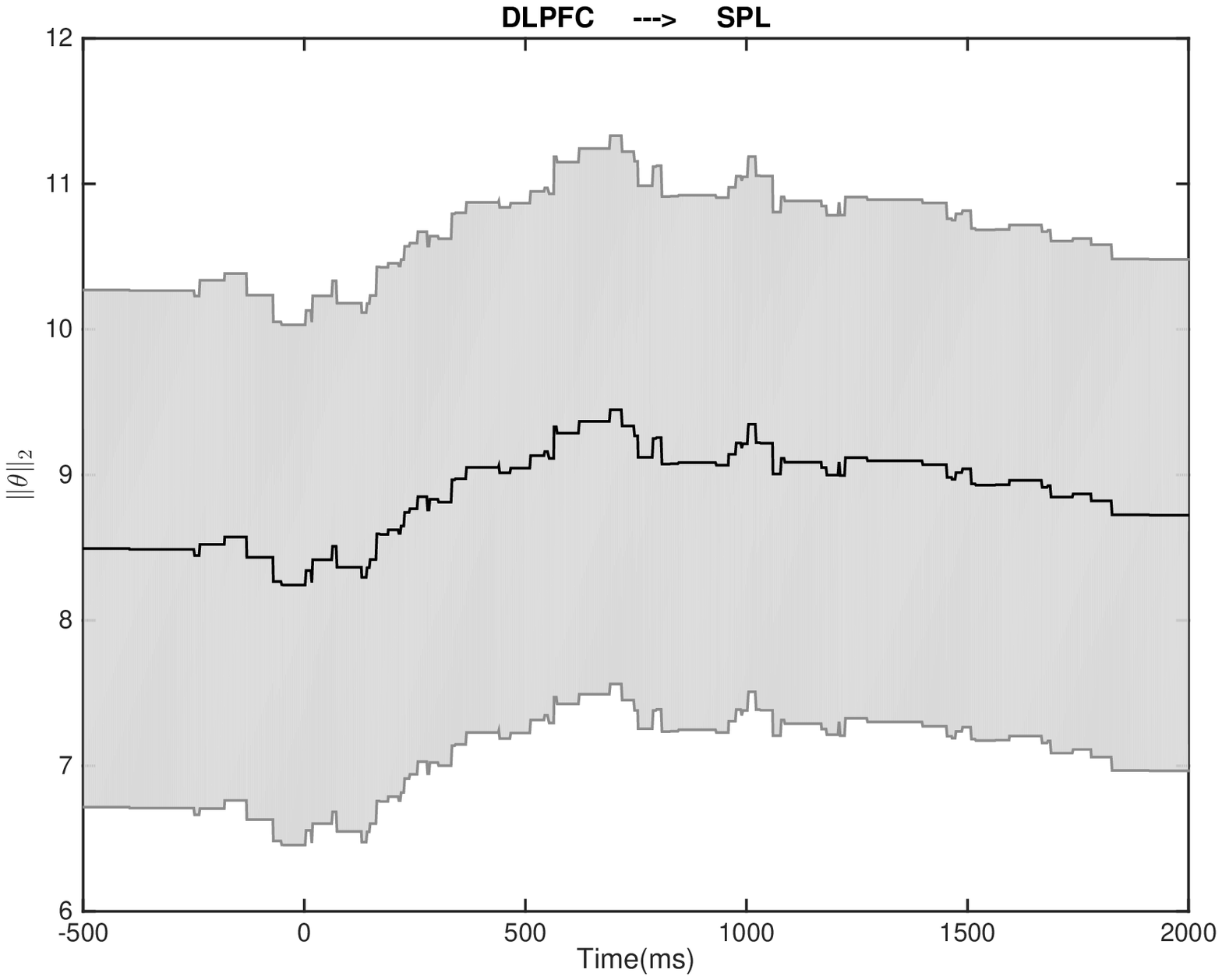}
            \caption[]%
            {{\small $\ell_2$ norm of edge DLPFC $\rightarrow$ SPL.}}    
            \label{fig:dlpfc-spl}
        \end{subfigure}
        \caption[ ]
        {\small $\ell_2$ norms of coefficients between pairs of time series in the frontoparietal region.} 
        \label{fig:Coef_FP}
    \end{figure}
\end{center}

As an aside, we note that comparatively few interactions were found between the visual processing region and the frontoparietal region using our method (results not shown).
    
\section{Conclusion}
\noindent
Motivated by the types of questions arising in task-based neuroscience -- particularly using imaging modalities with fine-scale temporal resolution -- we proposed a novel method for simulataneous network inference and change point detection.  Various extensions are possible.  For example, a penalty in the spirit of the fused-lasso would be of interest here, to encourage a certain notion of temporal contiguity.  In addition, a speed-up of the implementation (particularly for the non-dyadic case) would be desirable -- and, indeed, necessary for larger networks than those studied here -- adopting, for example, ideas like those underlying the PELT algorithm presented by \cite{killick2012optimal}.  Finally, it would be natural to explore the utility of our proposed method in the context of  financial economics. 

%\textcolor{red}{\it [EK:  Additional discussion topics:  Larger changepoint literature?  Particularly with connections to multiscale methods, such as in the spirit of what Piotr and colleagues have done.  Financial applications, but challenges with data.  Fused-lasso possibility for inducing a sense of smoothness across time.  Also, point to multiscale and wavelets on graph, starting with work with Crovella, and the now quite active field of graph signal processing.  Ideally one might do a network by time multiscale representation if desired.}
\section{Acknowledgements}
\noindent
We would like to thank Lucia Vaina and Kunjan Rana for providing the MEG data and offering helpful discussion throughout. This work was supported in part by funding under AFOSR award 12RSL042 and NIH award 1R01NS095369-01.
%\newpage
\bibliographystyle{rss}
\bibliography{xinyu24}

\begin{thebibliography}{35}
\expandafter\ifx\csname natexlab\endcsname\relax\def\natexlab#1{#1}\fi
\expandafter\ifx\csname url\endcsname\relax
  \def\url#1{\texttt{#1}}\fi
\expandafter\ifx\csname urlprefix\endcsname\relax\def\urlprefix{URL: }\fi

\bibitem[{Amano et~al.(2012)Amano, Takeda, Haji, Terao, Maruya, Matsumoto,
  Murakami and Nishida}]{amano2012human}
Amano, K., Takeda, T., Haji, T., Terao, M., Maruya, K., Matsumoto, K.,
  Murakami, I. and Nishida, S. (2012) Human neural responses involved in
  spatial pooling of locally ambiguous motion signals.
\newblock \textit{Journal of neurophysiology}, \textbf{107}, 3493--3508.

\bibitem[{Bach(2008)}]{bach2008consistency}
Bach, F.~R. (2008) Consistency of the group lasso and multiple kernel learning.
\newblock \textit{The Journal of Machine Learning Research}, \textbf{9},
  1179--1225.

\bibitem[{Barigozzi and Brownlees(2014)}]{barigozzi2014nets}
Barigozzi, M. and Brownlees, C.~T. (2014) Nets: network estimation for time
  series.
\newblock \textit{Available at SSRN 2249909}.

\bibitem[{Basu et~al.(2015)Basu, Shojaie and Michailidis}]{JMLR:v16:basu15a}
Basu, S., Shojaie, A. and Michailidis, G. (2015) Network granger causality with
  inherent grouping structure.
\newblock \textit{Journal of Machine Learning Research}, \textbf{16}, 417--453.
\newblock \urlprefix\url{http://jmlr.org/papers/v16/basu15a.html}.

\bibitem[{Betancourt et~al.(2017)Betancourt, Rodr{\'\i}guez and
  Boyd}]{betancourt2017bayesian}
Betancourt, B., Rodr{\'\i}guez, A. and Boyd, N. (2017) Bayesian fused lasso
  regression for dynamic binary networks.
\newblock \textit{Journal of Computational and Graphical Statistics}.

\bibitem[{Bettencourt and Xu(2016)}]{bettencourt2016decoding}
Bettencourt, K.~C. and Xu, Y. (2016) Decoding the content of visual short-term
  memory under distraction in occipital and parietal areas.
\newblock \textit{Nature neuroscience}, \textbf{19}, 150--157.

\bibitem[{Bolstad et~al.(2011)Bolstad, Van~Veen and Nowak}]{bolstad2011causal}
Bolstad, A., Van~Veen, B.~D. and Nowak, R. (2011) Causal network inference via
  group sparse regularization.
\newblock \textit{Signal Processing, IEEE Transactions on}, \textbf{59},
  2628--2641.

\bibitem[{Boyd et~al.(2011)Boyd, Parikh, Chu, Peleato and
  Eckstein}]{boyd2011distributed}
Boyd, S., Parikh, N., Chu, E., Peleato, B. and Eckstein, J. (2011) Distributed
  optimization and statistical learning via the alternating direction method of
  multipliers.
\newblock \textit{Foundations and Trends{\textregistered} in Machine Learning},
  \textbf{3}, 1--122.

\bibitem[{Braddick et~al.(2000)Braddick, O?Brien, Wattam-Bell, Atkinson and
  Turner}]{braddick2000form}
Braddick, O., O?Brien, J., Wattam-Bell, J., Atkinson, J. and Turner, R. (2000)
  Form and motion coherence activate independent, but not dorsal/ventral
  segregated, networks in the human brain.
\newblock \textit{Current Biology}, \textbf{10}, 731--734.

\bibitem[{Bullmore and Sporns(2009)}]{bullmore2009complex}
Bullmore, E. and Sporns, O. (2009) Complex brain networks: graph theoretical
  analysis of structural and functional systems.
\newblock \textit{Nature Reviews Neuroscience}, \textbf{10}, 186--198.

\bibitem[{Calabro and Vaina(2012)}]{calabro2012interaction}
Calabro, F. and Vaina, L. (2012) Interaction of cortical networks mediating
  object motion detection by moving observers.
\newblock \textit{Experimental brain research}, \textbf{221}, 177--189.

\bibitem[{Davis et~al.(2008)Davis, Lee and Rodriguez-Yam}]{davis2008break}
Davis, R.~A., Lee, T. and Rodriguez-Yam, G.~A. (2008) Break detection for a
  class of nonlinear time series models.
\newblock \textit{Journal of Time Series Analysis}, \textbf{29}, 834--867.

\bibitem[{DeVore(1998)}]{devore1998nonlinear}
DeVore, R.~A. (1998) Nonlinear approximation.
\newblock \textit{Acta numerica}, \textbf{7}, 51--150.

\bibitem[{Donoho(1993)}]{donoho1993unconditional}
Donoho, D.~L. (1993) Unconditional bases are optimal bases for data compression
  and for statistical estimation.
\newblock \textit{Applied and computational harmonic analysis}, \textbf{1},
  100--115.

\bibitem[{Donoho(1997)}]{donoho1997cart}
--- (1997) Cart and best-ortho-basis: a connection.
\newblock \textit{Ann. Statist.}, \textbf{25}, 1870--1911.

\bibitem[{Fouque et~al.(2011)Fouque, Papanicolaou, Sircar and
  S{\o}lna}]{fouque2011multiscale}
Fouque, J.-P., Papanicolaou, G., Sircar, R. and S{\o}lna, K. (2011)
  \textit{Multiscale stochastic volatility for equity, interest rate, and
  credit derivatives}.
\newblock Cambridge University Press.

\bibitem[{Granger(1969)}]{granger1969investigating}
Granger, C.~W. (1969) Investigating causal relations by econometric models and
  cross-spectral methods.
\newblock \textit{Econometrica}, 424--438.

\bibitem[{Hamilton(1983)}]{hamilton1983oil}
Hamilton, J.~D. (1983) Oil and the macroeconomy since world war ii.
\newblock \textit{The Journal of Political Economy}, 228--248.

\bibitem[{Hiemstra and Jones(1994)}]{hiemstra1994testing}
Hiemstra, C. and Jones, J.~D. (1994) Testing for linear and nonlinear granger
  causality in the stock price-volume relation.
\newblock \textit{The Journal of Finance}, \textbf{49}, 1639--1664.

\bibitem[{Honey et~al.(2007)Honey, K{\"o}tter, Breakspear and
  Sporns}]{honey2007network}
Honey, C.~J., K{\"o}tter, R., Breakspear, M. and Sporns, O. (2007) Network
  structure of cerebral cortex shapes functional connectivity on multiple time
  scales.
\newblock \textit{Proc. Natn. Acad. Sci. USA}, \textbf{104}, 10240--10245.

\bibitem[{Killick et~al.(2012)Killick, Fearnhead and
  Eckley}]{killick2012optimal}
Killick, R., Fearnhead, P. and Eckley, I.~A. (2012) Optimal detection of
  changepoints with a linear computational cost.
\newblock \textit{Journal of the American Statistical Association},
  \textbf{107}, 1590--1598.

\bibitem[{Kolaczyk(2009)}]{Kolaczyk:2009:SAN:1593430}
Kolaczyk, E.~D. (2009) \textit{Statistical Analysis of Network Data: Methods
  and Models}.
\newblock Springer Publishing Company, Incorporated, 1st edn.

\bibitem[{Kolaczyk and Nowak(2005)}]{kolaczyk2005multiscale}
Kolaczyk, E.~D. and Nowak, R.~D. (2005) Multiscale generalised linear models
  for nonparametric function estimation.
\newblock \textit{Biometrika}, \textbf{92}, 119--133.

\bibitem[{Li and Barron(2000)}]{li2000mixture}
Li, J.~Q. and Barron, A.~R. (2000) Mixture density estimation.
\newblock In \textit{Advances in neural information processing systems},
  279--285.

\bibitem[{Long et~al.(2005)Long, Brown, Triantafyllou, Aharon, Wald and
  Solo}]{long2005nonstationary}
Long, C., Brown, E., Triantafyllou, C., Aharon, I., Wald, L. and Solo, V.
  (2005) Nonstationary noise estimation in functional mri.
\newblock \textit{NeuroImage}, \textbf{28}, 890--903.

\bibitem[{Louie and Kolaczyk(2006)}]{louie2006multiscale}
Louie, M.~M. and Kolaczyk, E.~D. (2006) A multiscale method for disease mapping
  in spatial epidemiology.
\newblock \textit{Statistics in medicine}, \textbf{25}, 1287--1306.

\bibitem[{Mallat(1989)}]{mallat1989theory}
Mallat, S.~G. (1989) A theory for multiresolution signal decomposition: the
  wavelet representation.
\newblock \textit{Pattern Analysis and Machine Intelligence, IEEE Transactions
  on}, \textbf{11}, 674--693.

\bibitem[{Meinshausen and B{\"u}hlmann(2006)}]{meinshausen2006high}
Meinshausen, N. and B{\"u}hlmann, P. (2006) High-dimensional graphs and
  variable selection with the lasso.
\newblock \textit{Ann. Statist.}, 1436--1462.

\bibitem[{Mukhopadhyay and Chatterjee(2007)}]{mukhopadhyay2007causality}
Mukhopadhyay, N.~D. and Chatterjee, S. (2007) Causality and pathway search in
  microarray time series experiment.
\newblock \textit{Bioinformatics}, \textbf{23}, 442--449.

\bibitem[{M{\"u}ller(2001)}]{muller2001stochastic}
M{\"u}ller, A. (2001) Stochastic ordering of multivariate normal distributions.
\newblock \textit{Annals of the Institute of Statistical Mathematics},
  \textbf{53}, 567--575.

\bibitem[{Rana and Vaina(2014)}]{rana2014functional}
Rana, K.~D. and Vaina, L.~M. (2014) Functional roles of 10 hz alpha-band power
  modulating engagement and disengagement of cortical networks in a complex
  visual motion task.
\newblock \textit{PloS one}, \textbf{9}, e107715.

\bibitem[{Sims(1972)}]{sims1972money}
Sims, C.~A. (1972) Money, income, and causality.
\newblock \textit{The American Economic Review}, \textbf{62}, 540--552.

\bibitem[{Willett and Nowak(2007)}]{willett2007multiscale}
Willett, R.~M. and Nowak, R.~D. (2007) Multiscale poisson intensity and density
  estimation.
\newblock \textit{IEEE Transactions on Information Theory}, \textbf{53},
  3171--3187.

\bibitem[{Yuan and Lin(2006)}]{yuan2006model}
Yuan, M. and Lin, Y. (2006) Model selection and estimation in regression with
  grouped variables.
\newblock \textit{J. R. Statist. Soc. B}, \textbf{68}, 49--67.

\bibitem[{Zhao and Yu(2006)}]{zhao2006model}
Zhao, P. and Yu, B. (2006) On model selection consistency of lasso.
\newblock \textit{The Journal of Machine Learning Research}, \textbf{7},
  2541--2563.

\end{thebibliography}
\nocite{*}
\newpage
\section{Appendix}
\subsection{Algorithm using RDP}
\noindent
Here we provide the algorithm for implementation based on recursive dyadic partitions. Assume the length of the time series equals $T = 2^J$ and $j_{min} = \min_j$ such that $2^j > p+1$. Note that $p+1$ is the minimum required number of observations to fit the restricted VAR(p) model. Assume $J > j_{min}$,\\
\begin{algorithm}[H]
 \KwData{$\X(u)$, $\X(-u)$, $p$}
 \KwResult{$\boldsymbol{\hat\theta}_{RDP}$}
 \For{i = $0:2^{(J-j_{min})}-1$}{
  		Fit restricted VAR(p) model for $x_I(u)$,  for $I = \{t: t \in [2^{j_{min}}*i +1, 2^{j_{min}} *(i+1) ]\}$
  		Compute and store $pl_I$ on each interval $I$\;
  		optimumModel $\gets pl_I$\;
 }
  \For{j = $J-j_{min}-1:0$}{
 	\For{i = $0:2^j-1$}{
  		Fit restricted VAR(p) model for $\X_I(u)$, for $I = \{t: t \in [2^{(J-j)}*i+1, 2^{(J-j)}*(i+1)] \}$\;
  		Compute and store $pl_I$ on each interval $I$\;
  		\eIf{$pl_I \leq pl_{I_l^i} + pl_{I_r^i} + \text{Penalty}$}{
  		 optimumModel $\gets pl_I$\;
  		 Update changePoint\;
   		}{
   		optimumModel $\gets pl_l$ and $pl_r$\;
   		Update changePoint\;
  	}
  }
 }
 \caption{Multiscale dynamic causal network using RDP}
 \label{algorithm.RDP}
\end{algorithm}

Algorithm \ref{algorithm.RDP} splits only at dyadic positions. The candidate partitions $\PP \preceq \PP_{D_y}^*$ can be represented as subtrees of a binary tree of depth $\log_2 T$. Given a dataset of length $T = 2^J$, we have $2^0$ root node, $2^1$ nodes at level $1$, $2^2$ nodes, $2^3$ nodes, and so on, at the following levels, until we reach the leaf level, which has $2^{(J-1)}$ nodes. The complexity of the algorithm is then of order $\mathcal{O}(T)$ calls to fit the group lasso regression and $\mathcal{O}(T)$ calls for comparisons.

\subsection{Proof of theorem \ref{splitting}}
\begin{proof} Theorem \ref{splitting}\\
The proof contains two parts. In the first part, we show that equation (\ref{equa1}) holds, under $H_0$. In the second part, we show that equation (\ref{equa2}) holds, under $H_1$.

\noindent \textit{Part 1}\\
We begin by defining the group lasso penalized likelihood on an interval $I$:
\begin{align}
PL_I = \frac{1}{|I|}\left\|\X_I(u) -\X_I(-u)\boldsymbol\theta_I(u,v) \right\|_2^2 + \lambda_I  \sum_{v \in V\backslash \{u\}}\left\|\boldsymbol\theta_I(u,v)\right\|_2.
\label{pl}
\end{align}
Let $\boldsymbol{\hat\theta}_{1:T}$ be the $\boldsymbol\theta$ that minimizes the penalized likelihood (\ref{pl}) on the interval from $1$ to $T$ and $\hat{PL}_{1:T}$ be the quantity upon substituting $\boldsymbol{\hat\theta}_{1:T}$ in equation (\ref{pl}). Consider any alternative model with a change point detected at point $\hat\tau\in (1, T)$. Denote by $\boldsymbol{\hat\theta}_{1:\hat\tau}$ and $\boldsymbol{\hat\theta}_{\hat\tau:T}$ the coefficients $\boldsymbol\theta$ that minimize equation (\ref{pl}) over intervals  $[1, \hat\tau]$ and $(\hat\tau, T]$, respectively.
Given our model, equation (\ref{equa1}) in theorem \ref{splitting} is equivalent to
\begin{align*}
\mathbb{P}_{H_0}(\hat{PL}_{1:T} \leq \hat{PL}_{1:\hat\tau} + \hat{PL}_{\hat\tau:T} + C_3\log T) \longrightarrow 1.
\end{align*}
The additional term $C_3\log T$ comes from the fact that the alternative model has 1 more partition than the null model, with $C_3 = 1/2$ using RDP and $C_3 = 3/2$ using RP.  We expand $\hat{PL}_{1:\hat\tau} + \hat{PL}_{\hat\tau:T} - \hat{PL}_{1:T} +C_3\log T$ and get:
\begin{align}
& \frac{1}{\hat\tau}\left\|\X_{1:\hat\tau}(u) - \sum_{v\in V\backslash\{u\}}\X_{1:\hat\tau}(v)\boldsymbol{\hat{\theta}}_{1:\hat\tau}(u,v) \right\|_2^2 + \lambda_{1:\hat\tau}\sum_{v\in V\backslash\{u\}}\left\|\boldsymbol{\hat{\theta}}_{1:\hat\tau}(u,v)\right\|_2 \nonumber \\
+ \quad & \frac{1}{T- \hat\tau}\left\|\X_{\hat\tau:T}(u) - \sum_{v\in V\backslash\{u\}}\X_{\hat\tau:T}(v)\boldsymbol{\hat{\theta}}_{\hat\tau:T}(u,v) \right\|_2^2 + \lambda_{\hat\tau:T}\sum_{v\in V\backslash\{u\}}\left\|\boldsymbol{\hat{\theta}}_{\hat\tau:T}(u,v)\right\|_2 \nonumber \\
- \quad & \frac{1}{T}\left\|\X_{1:T}(u) - \sum_{v\in V\backslash\{u\}}\X_{1:T}(v)\boldsymbol{\hat{\theta}}_{1:T}(u,v) \right\|_2^2 - \lambda_{1:T}\sum_{v\in V\backslash\{u\}}\left\|\boldsymbol{\hat{\theta}}_{1:T}(u,v)\right\|_2 + C_3\log T.
\label{form1}
\end{align}
By rewriting the last line of equation (\ref{form1}), we have
\begin{align}
& \frac{1}{\hat\tau}\left\|\X_{1:\hat\tau}(u) - \sum_{v\in V\backslash\{u\}}\X_{1:\hat\tau}(v)\boldsymbol{\hat{\theta}}_{1:\hat\tau}(u,v) \right\|_2^2 + \lambda_{1:\hat\tau}\sum_{v\in V\backslash\{u\}}\left\|\boldsymbol{\hat{\theta}}_{1:\hat\tau}(u,v)\right\|_2 \nonumber \\
+ \quad & \frac{1}{T- \hat\tau}\left\|\X_{\hat\tau:T}(u) - \sum_{v\in V\backslash\{u\}}\X_{\hat\tau:T}(v)\boldsymbol{\hat{\theta}}_{\hat\tau:T}(u,v) \right\|_2^2 + \lambda_{\hat\tau:T}\sum_{v\in V\backslash\{u\}}\left\|\boldsymbol{\hat{\theta}}_{\hat\tau:T}(u,v)\right\|_2 \nonumber \\
- \quad & \frac{1}{T}\left\|\X_{1:\hat\tau}(u) - \sum_{v\in V\backslash\{u\}}\X_{1:\hat\tau}(v)\boldsymbol{\hat{\theta}}_{1:T}(u,v) \right\|_2^2   \nonumber\\
- \quad&  \frac{1}{T}\left\|\X_{\hat\tau :T}(u) - \sum_{v\in V\backslash\{u\}}\X_{\hat\tau :T}(v)\boldsymbol{\hat{\theta}}_{1:T}(u,v) \right\|_2^2\nonumber \\
- \quad & \lambda_{1:T}\sum_{v\in V\backslash\{u\}}\left\|\boldsymbol{\hat{\theta}}_{1:T}(u,v)\right\|_2 + C_3\log T.
\label{form2}
\end{align}
We then add and subtract a term in both line 3 and line 4  of equation (\ref{form2}). In doing so, we have:
{\small
\begin{align}
& \frac{1}{\hat\tau}\left\|\X_{1:\hat\tau}(u) - \sum_{v\in V\backslash\{u\}}\X_{1:\hat\tau}(v)\boldsymbol{\hat{\theta}}_{1:\hat\tau}(u,v) \right\|_2^2 + \lambda_{1:\hat\tau}\sum_{v\in V\backslash\{u\}}\left\|\boldsymbol{\hat{\theta}}_{1:\hat\tau}(u,v)\right\|_2 \nonumber \\
+ \quad & \frac{1}{T- \hat\tau}\left\|\X_{\hat\tau:T}(u) - \sum_{v\in V\backslash\{u\}}\X_{\hat\tau:T}(v)\boldsymbol{\hat{\theta}}_{\hat\tau:T}(u,v) \right\|_2^2 + \lambda_{\hat\tau:T}\sum_{v\in V\backslash\{u\}}\left\|\boldsymbol{\hat{\theta}}_{\hat\tau:T}(u,v)\right\|_2 \nonumber \\
- \quad & \frac{1}{T}\left\|\X_{1:\hat\tau}(u) - \sum_{v\in V\backslash\{u\}}\X_{1:\hat\tau}(v)\boldsymbol{\hat{\theta}}_{1:\hat\tau}(u,v) + \sum_{v\in V\backslash\{u\}}\X_{1:\hat\tau}(v)\boldsymbol{\hat{\theta}}_{1:\hat\tau}(u,v)- \sum_{v\in V\backslash\{u\}}\X_{1:\hat\tau}(v)\boldsymbol{\hat{\theta}}_{1:T}(u,v) \right\|_2^2   \nonumber\\
- \quad&  \frac{1}{T}\left\|\X_{\hat\tau :T}(u) - \sum_{v\in V\backslash\{u\}}\X_{\hat\tau:T}(v)\boldsymbol{\hat{\theta}}_{\hat\tau:T}(u,v) + \sum_{v\in V\backslash\{u\}}\X_{\hat\tau:T}(v)\boldsymbol{\hat{\theta}}_{\hat\tau:T}(u,v) - \sum_{v\in V\backslash\{u\}}\X_{\hat\tau :T}(v)\boldsymbol{\hat{\theta}}_{1:T}(u,v) \right\|_2^2\nonumber \\
- \quad & \lambda_{1:T}\sum_{v\in V\backslash\{u\}}\left\|\boldsymbol{\hat{\theta}}_{1:T}(u,v)\right\|_2 + C_3\log T.
\label{form3}
\end{align}
}
From which we have that:
{\small
\begin{align}
\text{equation (\ref{form3})} \nonumber\\
\geq \quad & \frac{1}{\hat\tau}\left\|\X_{1:\hat\tau}(u) - \sum_{v\in V\backslash\{u\}}\X_{1:\hat\tau}(v)\boldsymbol{\hat{\theta}}_{1:\hat\tau}(u,v) \right\|_2^2 \nonumber \\
+ \quad & \frac{1}{T- \hat\tau}\left\|\X_{\hat\tau:T}(u) - \sum_{v\in V\backslash\{u\}}\X_{\hat\tau:T}(v)\boldsymbol{\hat{\theta}}_{\hat\tau:T}(u,v) \right\|_2^2 \nonumber \\
- \quad & \frac{1}{T}\left\|\X_{1:\hat\tau}(u) - \sum_{v\in V\backslash\{u\}}\X_{1:\hat\tau}(v)\boldsymbol{\hat{\theta}}_{1:\hat\tau}(u,v) \right\|_2^2 \nonumber \\
- \quad & \frac{1}{T}\left \| \sum_{v\in V\backslash\{u\}}\X_{1:\hat\tau}(v)\boldsymbol{\hat{\theta}}_{1:\hat\tau}(u,v)- \sum_{v\in V\backslash\{u\}}\X_{1:\hat\tau}(v)\boldsymbol{\hat{\theta}}_{1:T}(u,v) \right\|_2^2   \nonumber\\
-\quad& \frac{2}{T}\left(\left\|\X_{1:\hat\tau}(u) - \sum_{v\in V\backslash\{u\}}\X_{1:\hat\tau}(v)\boldsymbol{\hat{\theta}}_{1:\hat\tau}(u,v) \right\|_2 \right.\nonumber \\
&\times \left. \left\| \sum_{v\in V\backslash\{u\}}\X_{1:\hat\tau}(v)\boldsymbol{\hat{\theta}}_{1:\hat\tau}(u,v)- \sum_{v\in V\backslash\{u\}}\X_{1:\hat\tau}(v)\boldsymbol{\hat{\theta}}_{1:T}(u,v) 	\right\|_2 \right) \nonumber\\
- \quad&  \frac{1}{T}\left\|\X_{\hat\tau :T}(u) - \sum_{v\in V\backslash\{u\}}\X_{\hat\tau:T}(v)\boldsymbol{\hat{\theta}}_{\hat\tau:T}(u,v) \right \|_2^2 \nonumber\\ 
- \quad &  \frac{1}{T}\left\| \sum_{v\in V\backslash\{u\}}\X_{\hat\tau:T}(v)\boldsymbol{\hat{\theta}}_{\hat\tau:T}(u,v) - \sum_{v\in V\backslash\{u\}}\X_{\hat\tau :T}(v)\boldsymbol{\hat{\theta}}_{1:T}(u,v) \right\|_2^2\nonumber \\
- \quad & \frac{2}{T}\left(\left\|\X_{\hat\tau :T}(u) - \sum_{v\in V\backslash\{u\}}\X_{\hat\tau:T}(v)\boldsymbol{\hat{\theta}}_{\hat\tau:T}(u,v) \right \|_2 \right. \nonumber \\
& \times \left. \left\|\sum_{v\in V\backslash\{u\}}\X_{\hat\tau:T}(v)\boldsymbol{\hat{\theta}}_{\hat\tau:T}(u,v) - \sum_{v\in V\backslash\{u\}}\X_{\hat\tau :T}(v)\boldsymbol{\hat{\theta}}_{1:T}(u,v) \right\|_2 \right) \nonumber \\
+ \quad &  \lambda_{\hat\tau:T}\sum_{v\in V\backslash\{u\}}\left\|\boldsymbol{\hat{\theta}}_{\hat\tau:T}(u,v)\right\|_2+  \lambda_{1:\hat\tau}\sum_{v\in V\backslash\{u\}}\left\|\boldsymbol{\hat{\theta}}_{1:\hat\tau}(u,v)\right\|_2 - \lambda_{1:T}\sum_{v\in V\backslash\{u\}}\left\|\boldsymbol{\hat{\theta}}_{1:T}(u,v)\right\|_2 \nonumber \\
+ \quad & C_3\log T.
\label{form4}
\end{align}
}
Under assumptions (\ref{A1}) to (\ref{A5}), \cite{bach2008consistency} reformulated the group lasso penalized likelihood (\ref{pl}) as:
\begin{align}
PL_I =  \hat{\boldsymbol\Sigma}_{\X(u)\X(u)} - 2\hat{\boldsymbol\Sigma}_{\X(-u)\X(u)}^\prime\boldsymbol{\theta} + \boldsymbol{\theta}^\prime\boldsymbol{\hat{\Sigma}}_{\X(-u)\X(-u)}\boldsymbol\theta + \lambda_I \sum_{v\in V\backslash\{u\}}\left\|\boldsymbol{\theta}(u,v)\right\|_2
\label{alternateform}
\end{align}
where $\hat{\boldsymbol\Sigma}_{\X(u)\X(u)} = \frac{1}{|I|}\X(u)^\prime \Pi_{|I|}\X(u)$, $\hat{\boldsymbol\Sigma}_{\X(-u) \X(u)} = \frac{1}{|I|}\X(-u)^\prime\Pi_{|I|}\X(u)$ and\\ $\boldsymbol\theta^\prime\hat{\boldsymbol\Sigma}_{\X(-u)\X(-u)}\boldsymbol\theta = \frac{1}{|I|}\X(-u)^\prime\Pi_{|I|}\X(-u)$ are the empirical covariance matrices with $\Pi_{|I|}$ defined as $\Pi_{|I|} = \mathbf{I}_{|I|}-\frac{1}{|I|}\mathbf{1}_{|I|}\mathbf{1}_{|I|}^\prime$ and showed that the group lasso estimator $\boldsymbol{\hat\theta}$ converges in probability to $\boldsymbol{\theta}$. Using expression in (\ref{alternateform}) and collecting similar terms, we could then rewrite (\ref{form4}) as:
{\fontsize{8}{8}
\begin{align}
& \frac{T-\hat\tau}{T}\left\{\hat{\boldsymbol\Sigma}_{\X_{1:\hat\tau}(u) \X_{1:\hat\tau}(u)} - 2\hat{\boldsymbol\Sigma}_{\X_{1:\hat\tau}(-u) \X_{1:\hat\tau}(u)} \hat{\boldsymbol\theta}_{1:\hat\tau} +  \hat{\boldsymbol\theta}_{{1:\hat\tau}}^\prime \hat{\boldsymbol\Sigma}_{\X_{1:\hat\tau}(-u)\X_{1:\hat\tau}(-u)} \hat{\boldsymbol\theta}_{1:\hat\tau} \right\} \nonumber\\ 
+\quad & \frac{\hat\tau}{T}\left\{\hat{\boldsymbol\Sigma}_{\X_{\hat\tau:T}(u) \X_{\hat\tau:T}(u)} - 2\hat{\boldsymbol\Sigma}_{\X_{\hat\tau:T}(-u) \X_{\hat\tau:T}(u)} \hat{\boldsymbol\theta}_{\hat\tau:T} +  \hat{\boldsymbol\theta}_{{\hat\tau:T}}^\prime \hat{\boldsymbol\Sigma}_{\X_{\hat\tau:T}(-u)\X_{1:\hat\tau}(-u)} \hat{\boldsymbol\theta}_{\hat\tau:T} \right\} \label{partt1}\\
- \quad & \left\| \hat{\boldsymbol\Sigma}^{1/2}_{\X_{1:\hat\tau}(-u) \X_{1:\hat\tau}(-u)} \left( \hat{\boldsymbol\theta}_{1:\hat\tau} -  \hat{\boldsymbol\theta}_{1:T} \right)\right\|_2^2
-  \left\| \hat{\boldsymbol\Sigma}^{1/2}_{\X_{\hat\tau:T}(-u) \X_{\hat\tau:T}(-u)} \left( \hat{\boldsymbol\theta}_{\hat\tau:T} -  \hat{\boldsymbol\theta}_{1:T} \right)\right\|_2^2\label{partt2} \\ 
- \quad & \frac{2}{T}\left(\left\|\X_{1:\hat\tau}(u) - \sum_{v\in V\backslash\{u\}}\X_{1:\hat\tau}(v)\boldsymbol{\hat{\theta}}_{1:\hat\tau}(u,v) \right\|_2 \left \| \sum_{v\in V\backslash\{u\}}\X_{1:\hat\tau}(v)\left(\boldsymbol{\hat{\theta}}_{1:\hat\tau}(u,v)- \boldsymbol{\hat{\theta}}_{1:T}(u,v)\right) \right\|_2\right) \label{partt4}\\
- \quad & \frac{2}{T}\left(\left\|\X_{\hat\tau:T}(u) - \sum_{v\in V\backslash\{u\}}\X_{\hat\tau:T}(v)\boldsymbol{\hat{\theta}}_{\hat\tau:T}(u,v) \right\|_2 \left \| \sum_{v\in V\backslash\{u\}}\X_{\hat\tau:T}(v)\left(\boldsymbol{\hat{\theta}}_{\hat\tau:T}(u,v)- \boldsymbol{\hat{\theta}}_{1:T}(u,v)\right) \right\|_2\right) \label{partt5}\\
+ \quad &\lambda_{1:\hat\tau}\sum_{v\in V\backslash\{u\}}\left\|\hat{\boldsymbol\theta}_{1:\hat\tau}(u,v)\right\|_2 +\lambda_{\hat\tau:T}\sum_{v\in V\backslash\{u\}}\left\|\hat{\boldsymbol\theta}_{\hat\tau:T}(u,v)\right\|_2 - \lambda_{1:T}\sum_{v\in V\backslash\{u\}}\left\|\hat{\boldsymbol\theta}_{1:T}(u,v)\right\|_2 + C_3\log T\label{partt3}.
\end{align}
}

Note that in the previous expression, the first two lines are by definition non-negative. The expression in the last line is composed of a collection of penalty terms. They are the group lasso penalties, and all of them converge to zero asymptotically assuming $\lambda_{(\cdot)} \longrightarrow 0$ and $\lambda_{(\cdot)} N \longrightarrow 0$. 

Since $\hat{\boldsymbol\theta}_{1:\hat\tau} \stackrel{P}{\longrightarrow} \boldsymbol\theta$ , $\hat{\boldsymbol\theta}_{\hat\tau:T} \stackrel{P}{\longrightarrow} \boldsymbol\theta$ and $\hat{\boldsymbol\theta}_{1:T} \stackrel{P}{\longrightarrow} \boldsymbol\theta$, $\hat{\boldsymbol\theta}_{1:\hat\tau} - \hat{\boldsymbol\theta}_{1:T} \stackrel{P}{\longrightarrow} 0$ and $X$'s have finite moments up to order 4, each term in \eqref{partt2}, \eqref{partt4} and \eqref{partt5} converges to 0 in probability.

Putting everything together, we then complete the proof of the first part of the theorem:
\begin{align*}
\mathbb{P}_{H_0}(\hat{PL}_{1:T} \leq \hat{PL}_{1:\hat\tau_i} + \hat{PL}_{\hat\tau:T} + C_3\log T) \longrightarrow 1.
\end{align*}
\noindent
\textit{Part 2}\\
Suppose $H_1$ is true. We denote the estimated change point by $\hat\tau$. We show that $\hat{PL}_{1:\hat\tau} + \hat{PL}_{\hat\tau:T}$ is minimized at $\hat\tau = \tau$. Assume we have a competing estimator $\tilde\tau$ with change point detected at time $\tilde\tau = s$ with $s \neq \tau$. We show that
\begin{align}
\hat{PL}_{1:\hat\tau} + \hat{PL}_{\hat\tau:T} \leq \hat{PL}_{1:s} +  \hat{PL}_{s:T} 
\label{object}
\end{align}
holds with high probability under $H_1$. 
Without loss of generality, we assume that $\tau - s = \delta$, for some $\delta > 0$ as shown in figure \ref{fig.2}. For the case that $s > \tau$, a similar argument holds.
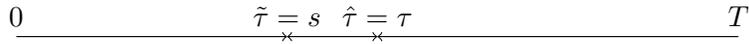
\begin{figure}[H]
\centering
\begin{tikzpicture}[scale=.6]
\draw [->] (0,0) -- (6,0)
node[pos=0,above] {$0$};
\draw [<->] (6,0) -- (8,0)
node[pos=0,above] {$\tilde\tau = s$};
\draw [<-] (8,0) -- (16,0)
node[pos=0,above] {$\hat\tau = \tau$}
node[pos=1,above] {$T$};
\end{tikzpicture}
\caption{Relative position of two detected change points}
\label{fig.2}
\end{figure}
Denote by $\boldsymbol{\hat\theta}_{1:\hat\tau}$ and $\boldsymbol{\hat\theta}_{\hat\tau:T}$ the estimated coefficients that minimize the penalized likelihoods, given that $I = \{t: t \in [1,\hat\tau)\}$ and $I = \{t: t \in [\hat\tau, T]\}$. We also define $\boldsymbol{\hat\theta}_{1:s}$ and $\boldsymbol{\hat\theta}_{s:T}$ to be the estimated coefficients that minimize the penalized likelihoods in \ref{pl}, given that $I = \{t: t \in [1,s)\}$ and $I = \{t: t \in [s, T]\}$.
The key idea is that $\hat{\boldsymbol\theta}_{1:\hat\tau}$ and $\hat{\boldsymbol\theta}_{\hat\tau:T}$ are consistent estimators of $\boldsymbol\theta_{1:\tau}$ and $\boldsymbol\theta_{\tau:T}$ but $\hat{\boldsymbol\theta}_{s:T}$ is not a consistent estimator of $\boldsymbol\theta_{1:\tau}$ nor $\boldsymbol\theta_{\tau:T}$ due to the mis-specification error. Therefore, one of the estimators from $\boldsymbol{\hat\theta}_{1:s}$ and $\boldsymbol{\hat\theta}_{s:T}$ such that $s < \tau$ is not a consistent estimator on the corresponding intervals. Formally, we have that
{\normalsize
\begin{align}
&\hat{PL}_{1:s} + \hat{PL}_{s:T} \nonumber \\ 
&= \dfrac{1}{s}\left\|\X_{1:s}(u) - \sum_{v\in V\backslash\{u\}}\X_{1:s}(v)\boldsymbol{\hat\theta}_{1:s}(u,v) \right\|_2^2 + \lambda_{1:s}\sum_{v\in V\backslash\{u\}}\left\|\boldsymbol{\hat\theta}_{1:s}(u,v)\right\|_2 \nonumber \\
&+ \dfrac{1}{T-s}\left\|\X_{s:T}(u) - \sum_{v\in V\backslash\{u\}}\X_{s:T}(v)\boldsymbol{\hat\theta}_{s:T}(u,v) \right\|_2^2 + \lambda_{s:T}\sum_{v\in V\backslash\{u\}}\left\|\boldsymbol{\hat\theta}_{s:T}(u,v)\right\|_2 \nonumber \\
&= \dfrac{1}{s}\left\|\X_{1:s}(u) - \sum_{v\in V\backslash\{u\}}\X_{1:s}(v)\boldsymbol{\hat\theta}_{1:s}(u,v) \right\|_2^2 + \lambda_{1:s}\sum_{v\in V\backslash\{u\}}\left\|\boldsymbol{\hat\theta}_{1:s}(u,v)\right\|_2 \nonumber\\
&+ \frac{1}{T-s}\left\|\X_{s:\tau}(u) - \sum_{v\in V\backslash\{u\}}\X_{s:\tau}(v)\boldsymbol{\hat\theta}_{s:T}(u,v) \right\|_2^2 + \frac{\delta \lambda_{s:T}}{T-s}\sum_{v\in V\backslash\{u\}}\left\|\boldsymbol{\hat\theta}_{s:T}(u,v)\right\|_2 \label{p1}\\
&+ \frac{1}{T-s}\left\|\X_{\tau:T}(u) - \sum_{v\in V\backslash\{u\}}\X_{\tau:T}(v)\boldsymbol{\hat\theta}_{s:T}(u,v) \right\|_2^2 + \frac{(T-s-\delta)\lambda_{s:T}}{T-s}\sum_{v\in V\backslash\{u\}}\left\|\boldsymbol{\hat\theta}_{1:s}(u,v)\right\|_2 \label{p2}
\end{align}
}
and
\begin{align*}
&\hat{PL}_{1:\hat\tau} + \hat{PL}_{\hat\tau:T} \\ 
&= \dfrac{1}{\tau}\left\|\X_{1:\hat\tau}(u) - \sum_{v\in V\backslash\{u\}}\X_{1:\hat\tau}(v)\boldsymbol{\hat\theta}_{1:\hat\tau}(u,v) \right\|_2^2 + \lambda_{1:\hat\tau}\sum_{v\in V\backslash\{u\}}\left\|\boldsymbol{\hat\theta}_{1:\hat\tau}(u,v)\right\|_2 \\
&+ \dfrac{1}{T-\tau}\left\|\X_{\hat\tau:T}(u) - \sum_{v\in V\backslash\{u\}}\X_{\hat\tau:T}(v)\boldsymbol{\hat\theta}_{\hat\tau:T}(u,v) \right\|_2^2 + \lambda_{\hat\tau:T}\sum_{v\in V\backslash\{u\}}\left\|\boldsymbol{\hat\theta}_{\hat\tau:T}(u,v)\right\|_2 
\end{align*}
We write expression (\ref{p1}) as $\hat{PL}_{1:s} + \hat{PL}_{s:\hat\tau}$, and expression (\ref{p2}), as $\tilde{PL}_{s:T}$. We show (\ref{object}) holds by first showing that $\hat{PL}_{1:s} + \hat{PL}_{s:\hat\tau} \geq \hat{PL}_{1:\hat\tau}$, and then showing $\tilde{PL}_{s:T} \geq \hat{PL}_{\hat\tau:T}$.
We first compute $\hat{PL}_{1:s} + \hat{PL}_{s:\hat\tau} - \hat{PL}_{1:\hat\tau}$:
\begin{align*}
&= \dfrac{1}{s}\left\|\X_{1:s}(u) - \sum_{v\in V\backslash\{u\}}v_{1:s}(v)\boldsymbol{\hat\theta}_{1:s}(u,v) \right\|_2^2 + \lambda_{1:s}\sum_{\X\in V\backslash\{u\}}\left\|\boldsymbol{\hat\theta}_{1:s}(u,v)\right\|_2 \\
&+ \frac{1}{T-s}\left\|\X_{s:\hat\tau}(u) - \sum_{v\in V\backslash\{u\}}\X_{s:\hat\tau}(v)\boldsymbol{\hat\theta}_{s:T}(u,v) \right\|_2^2 + \frac{\delta \lambda_{s:T}}{T-s}\sum_{v\in V\backslash\{u\}}\left\|\boldsymbol{\hat\theta}_{s:T}(u,v)\right\|_2 \\
&-\dfrac{1}{\tau}\left\|\X_{1:\hat\tau}(u) - \sum_{v\in V\backslash\{u\}}\X_{1:\hat\tau}(v)\boldsymbol{\hat\theta}_{1:\hat\tau}(u,v) \right\|_2^2 - \lambda_{1:\hat\tau}\sum_{v\in V\backslash\{u\}}\left\|\boldsymbol{\hat\theta}_{1:\hat\tau}(u,v)\right\|_2.
\end{align*}
Assuming there is another group-lasso estimator defined on the the interval between $s$ and $\hat\tau$, which is given by
\begin{align*}
\boldsymbol{\hat\theta}_{s:\hat\tau} = \argmin_{\boldsymbol{\theta}} \frac{1}{\hat\tau-s}\left\|\X_{s:\hat\tau}(u) - \sum_{v\in V\backslash\{u\}}\X_{s:\hat\tau}(v)\boldsymbol{\theta}_{s:\hat\tau}(u,v) \right\|_2^2 + \lambda_{s:\hat\tau}\sum_{v\in V\backslash\{u\}}\left\|\boldsymbol\theta_{s:\hat\tau}(u,v)\right\|_2.
\end{align*}
The estimator $\boldsymbol{\hat\theta}_{s:\hat\tau} $ is again a consistent estimator of $\boldsymbol{\theta}_{1:\hat\tau}$ and we have that:
\begin{align}
&\frac{1}{\hat\tau-s}\left\|\X_{s:\hat\tau}(u) - \sum_{v\in V\backslash\{u\}}\X_{s:\hat\tau}(v)\boldsymbol{\hat\theta}_{s:\hat\tau}(u,v) \right\|_2^2 + \lambda_{s:\hat\tau}\sum_{v\in V\backslash\{u\}}\left\|\boldsymbol{\hat\theta}_{s:\hat\tau}(u,v)\right\|_2 \\ 
&\leq \frac{1}{T-s}\left\|\X_{s:\hat\tau}(u) - \sum_{v\in V\backslash\{u\}}\X_{s:\hat\tau}(v)\boldsymbol{\hat\theta}_{s:T}(u,v) \right\|_2^2 + \frac{\delta \lambda_{s:T}}{T-s}\sum_{v\in V\backslash\{u\}}\left\|\boldsymbol{\hat\theta}_{s:T}(u,v)\right\|_2
\label{substitute}
\end{align}
These are directly implied by Theorem (2) in \citep{bach2008consistency} given that $\boldsymbol{\hat\theta}_{s:T}$ is not consistent in the $\ell_2$ sense of estimating $\boldsymbol\theta_{1:\hat\tau}$ whenever $s \neq \hat\tau$.
Given (\ref{substitute}), we have that
\begin{align*}
&\hat{PL}_{1:s} + \hat{PL}_{s:\hat\tau} - \hat{PL}_{1:\hat\tau} \\
&\geq \dfrac{1}{s}\left\|\X_{1:s}(u) - \sum_{v\in V\backslash\{u\}}\X_{1:s}(v)\boldsymbol{\hat\theta}_{1:s}(u,v) \right\|_2^2 + \lambda_{1:s}\sum_{v\in V\backslash\{u\}}\left\|\boldsymbol{\hat\theta}_{1:s}(u,v)\right\|_2 \\
&+ \frac{1}{\hat\tau-s}\left\|\X_{s:\hat\tau}(u) - \sum_{v\in V\backslash\{u\}}\X_{s:\hat\tau}(v)\boldsymbol{\hat\theta}_{s:\hat\tau}(u,v) \right\|_2^2 + \lambda_{s:\hat\tau}\sum_{v\in V\backslash\{u\}}\left\|\boldsymbol{\hat\theta}_{s:\hat\tau}(u,v)\right\|_2 \\
&-\dfrac{1}{\hat\tau}\left\|\X_{1:\hat\tau}(u) - \sum_{v\in V\backslash\{u\}}\X_{1:\hat\hat\tau}(v)\boldsymbol{\hat\theta}_{1:\hat\tau}(u,v) \right\|_2^2 - \lambda_{1:\hat\tau}\sum_{v\in V\backslash\{u\}}\left\|\boldsymbol{\hat\theta}_{1:\hat\tau}(u,v)\right\|_2 \\
\end{align*}
The same argument in Part 1 holds here and we have 
\begin{align*}
\mathbb{P}_{H_1}\left(\hat{PL}_{1:s} + \hat{PL}_{s:\hat\tau} \geq \hat{PL}_{1:\hat\tau}\right) \longrightarrow 1\enskip .
\end{align*}
Note that $\boldsymbol{\hat\theta}_{s:T}$ is not a consistent estimator of $\boldsymbol{\theta}_{\hat\tau:T}$ given the change point.  Therefore, similar to \ref{substitute}, we have 
\begin{align*}
 & \dfrac{1}{T-\hat\tau}\left\|\X_{\hat\tau:T}(u) - \sum_{v\in V\backslash\{u\}}\X_{\hat\tau:T}(v)\boldsymbol{\hat\theta}_{\hat\tau:T}(u,v) \right\|_2^2 + \lambda_{\hat\tau:T}\sum_{v\in V\backslash\{u\}}\left\|\boldsymbol{\hat\theta}_{\hat\tau:T}(u,v)\right\|_2 \\
 \leq \quad& \frac{1}{T-s}\left\|\X_{\hat\tau:T}(u) - \sum_{v\in V\backslash\{u\}}\X_{\hat\tau:T}(v)\boldsymbol{\hat\theta}_{s:T}(u,v) \right\|_2^2 + \frac{(T-s-\delta)\lambda_{s:T}}{T-s}\sum_{v\in V\backslash\{u\}}\left\|\boldsymbol{\hat\theta}_{1:s}(u,v)\right\|_2 \\
\end{align*}
and so
\begin{align*}
 \mathbb{P}_{H_1}\left(\tilde{PL}_{s:T} \geq \hat{PL}_{\hat\tau:T}\right) \longrightarrow 1 \enskip .
\end{align*}
Putting the two parts together, we have
\begin{align*}
\mathbb{P}_{H_1}\left(\hat{PL}_{1:s} + \hat{PL}_{s:T}\geq \hat{PL}_{1:\hat\tau} + \hat{PL}_{\hat\tau:T} \right) \longrightarrow 1
\end{align*}
for any $s < \hat\tau$.
\end{proof}
\subsection{Proof of theorem \ref{finite}}
\noindent
Under the assumption of stationarity, we could omit the time index in this section, that is $\boldsymbol{\theta} = \boldsymbol{\theta}_t, \,\, \forall t$. To show theorem \ref{finite}, we begin with the following lemma.
\begin{Lem}
Given $\boldsymbol{\theta} \in \mathbb{R}^{(N-1)p}$, let $G(\boldsymbol{\theta}(u,v))$ be a $p$-dimensional vector with elements
\begin{eqnarray}
	G(\boldsymbol{\theta}(u,v)) 
	&= -2T^{-1}\left(\X(v)^\prime(\X(u) - \sum_{v\in V\backslash\{u\}}\X(v)\boldsymbol{\theta}(u,v))\right).
	\label{diff2}
\end{eqnarray}
A vector $\boldsymbol{\hat{\theta}}$ with $\|\boldsymbol{\hat{\theta}}(u,v)\|_2 = 0$, $\forall\, v \in V\backslash \{u\}$ is a solution to the group lasso type of estimator iff for all $v\in V \backslash \{u\}$, $ G(\boldsymbol{\hat{\theta}}(u,v)) +  \lambda\D(\boldsymbol{\hat{\theta}}(u,v)) = \mathbf{0}$, where $\|\D(\boldsymbol{\hat{\theta}}(u,v))\|_2 = 1$ in the case of $\|\boldsymbol{\hat{\theta}}(u,v)\|_2 > 0$ and $\|\D(\boldsymbol{\hat{\theta}}(u,v) \|_2 < 1$ in the case of $\|\boldsymbol{\hat{\theta}}(u,v)\|_2 = 0$. 
\label{lemma.1}
\end{Lem}
\begin{proof} Lemma \ref{lemma.1}\\
Under KKT conditions, using subdifferential methods, the subdifferential of 
\begin{equation*}
	 \frac{1}{T}\left\|\X(u) - \sum_{v\in V\backslash\{u\}}\X(v)\boldsymbol{\theta}(u,v)\right\|^2 + \lambda\sum_{v\in V\backslash\{u\}}\left\|\boldsymbol{\theta}(u,v)\right\|_2
\end{equation*}
is given by $G(\boldsymbol{\theta}(u,v)) + \lambda \D(\boldsymbol{\hat{\theta}}(u,v))$, where $\|\D(\boldsymbol{\hat{\theta}}(u,v))\|_2 = 1$ if $\|\boldsymbol{\theta}(u,v)\|_2 > 0$ and $\|\D(\boldsymbol{\hat{\theta}}(u,v))\|_2 < 1$ if $\|\boldsymbol{\theta}(u,v)\|_2 = 0$. The lemma follows.
\end{proof}
We now proof theorem \ref{finite}.
\begin{proof}
Assuming that $\hat{C}_u \nsubseteq C_u$, there must exist at least one estimated edge that joins two nodes in two different connectivity components. Given the assumptions, we use similar arguments as in the proof of Theorem 3 in \cite{meinshausen2006high}. Hence we have
\begin{equation*}
	\mathbb{P}(\exists \, u \in V: \hat{C}_u \nsubseteq C_u) \leq N \max_{u \in V} \mathbb{P}(\exists\, v \in V \backslash C_u: v \in \hat{\text{ne}}_u)\, ,
\end{equation*}
where $\hat{\text{ne}}_u$ is the estimated neighborhood of node $u$ and  $v \in \hat{\text{ne}}_u$ means $\| \boldsymbol{\hat{\theta}}(u,v)\|_2 > 0$.

Let $\mathscr{E}$ be the event that
\begin{equation*}
	\max_{u\in V \backslash C_u} \left\|G \left(\boldsymbol{\hat{\theta}}(u,v)\right) \right\|_2^2 < \lambda^2.
\end{equation*}
Conditional on the event $\mathscr{E}$, $\boldsymbol{\hat{\theta}}$ is also a solution to the group lasso problem. As $\|\boldsymbol{\hat{\theta}}(u,v)\|_2 = 0$ for all $v \in V \backslash C_u$, it follows from lemma (\ref{lemma.1}) that $\|\boldsymbol{\hat{\theta}}(u,v)\|_2 = 0$ for all $v\in V \backslash C_u$. Hence
\begin{eqnarray*}
	\mathbb{P}(\exists\, v\in V \backslash C_u: \|\boldsymbol{\hat{\theta}}(u,v)\|_2 > 0) &\leq & 1 - \mathbb{P}(\mathscr{E})\\
	&=& P\left(\max_{v \in V \backslash C_u} \left\|G\left(\boldsymbol{\hat{\theta}}(u,v)\right) \right\|_2^2 \geq \lambda^2 \right).
\end{eqnarray*}
It is then sufficient to show that 
\begin{equation*}
	N^2 \max_{u \in V\text{, } v\in V \backslash C_u} \mathbb{P}\left(\left\|G(\boldsymbol{\hat{\theta}}(u,v))\right\|_2^2 \geq \lambda^2\right) \leq \alpha.
\end{equation*}
Note that now the $v$ and $C_u$ are in different connected components, which means that $\X(v)$ is conditionally independent of $\X(C_u)$. Hence, conditional on all $\X(C_u)$, we have
\begin{eqnarray*}
	\left\|G(\boldsymbol{\hat{\theta}}(u,v))\right\|_2^2 	&=& \left\|-2T^{-1}\left(\X(v)^\prime(\X(u) - \sum_{i \in C_u}\X(i)\boldsymbol{\hat\theta}(u,i))\right)\right\|_2^2\\
	&=& 4T^{-2}\left\|( \mathbf{\hat R}_1, \cdots, \mathbf{\hat R}_p )^\prime  \right\|_2^2 
\end{eqnarray*}
where $\mathbf{\hat R}_\ell =X_{-\ell}(v)^\prime \left(\X(u) -  \sum_{i \in C_u}\X(i)\boldsymbol{\hat{\theta}}(u,i)\right)$ is the remainder term and is independent of $\X(v)$, at all lags $\ell$, for $\ell = 1,\cdots, p$. It follows that the joint distribution
\begin{equation*}
	( \mathbf{\hat R}_1, \cdots,  \mathbf{\hat R}_p| \X(C_u)) \sim N(\mathbf{0}, \mathbf{\Omega})
\end{equation*}
for some covariance matrix $\mathbf{\Omega}$.
Note that this is a conditional distribution given $\X(C_u)$. Hence, in the expression of $\boldsymbol\Omega$, every term appearing with a suffix $u$ is constant and every term appearing with a suffix $v$ is a normalized random variable. This simplifies the covariance term. 
Note that
\begin{equation*}
	{\boldsymbol\Omega}_{p \times p} = \textbf{Cov}\left(\mathbf{\hat R}_1, \cdots, \mathbf{\hat R}_p \right)
\end{equation*}
and
{\fontsize{9.5}{9.5}\selectfont 
\begin{align}
	&\textbf{tr}\left(\boldsymbol\Omega\right) 
	= \sum_{\ell=1}^p \textbf{Var} (\boldsymbol{\hat R}_\ell) = \sum_{\ell=1}^p \textbf{Var}\left(\sum_{t=1}^T\left(X_t(u)-\sum_{i\in C_u}X_{t-\ell}(i)\hat{\theta}^{(\ell)}(u,i)\right)X_{t-\ell}(v)\right)= \nonumber \\ &\sum_{\ell=1}^p\sum_{s=1}^T\sum_{t=1}^T\textbf{Cov}\left[\left(\left(X_t(u)-\sum_{i\in C_u}X_{t-\ell}(i)\hat{\theta}^{(\ell)}(u,i)\right)X_{t-\ell}(v)\right), \left(\left(X_s(u)-\sum_{i\in C_u}X_{s-\ell}(i)\hat{\theta}^{(\ell)}(u,i)\right)X_{s-\ell}(v)\right)\right]
	\label{trace}
\end{align}
}
Conditional on $\X(C_u)$, equation (\ref{trace}) can be further simplified as:
{\fontsize{8.5}{8.5}\selectfont 
\begin{align*}
	\textbf{tr}\left(\boldsymbol\Omega\right) 
	&=  \sum_{\ell=1}^p\sum_{t=1}^T\sum_{s=1}^T\left(X_t(u)-\sum_{i\in C_u}X_{t-\ell}(i)\hat{\theta}^{(\ell)}(u,i)\right) \left(X_s(u)-\sum_{i\in C_u}X_{s-\ell}(i)\hat{\theta}^{(\ell)}(u,i)\right) \textbf{Cov} \left[ X_{t-\ell}(v) ,X_{s-\ell}(v)  \right] \\
	&\leq \sum_{\ell=1}^p\sum_{t=1}^T\sum_{s=1}^T\left(X_t(u)-\sum_{i\in C_u}X_{t-\ell}(i)\hat{\theta}^{(\ell)}(u,i)\right) \left(X_s(u)-\sum_{i\in C_u}X_{s-\ell}(i)\hat{\theta}^{(\ell)}(u,i)\right)\sqrt{\textbf{Var}(X_{t-\ell}(v)) \textbf{Var}(X_{s-\ell}(v))}
\end{align*}
}
We have the above bounded by
\begin{align*}
	&\leq p \sum_{s=1}^T\sum_{t=1}^T  \left(X_t(u)-\sum_{i\in C_u}X_{t-\ell}(i)\hat{\theta}^{(\ell)}(u,i)\right) \left(X_s(u)-\sum_{i\in C_u}X_{s-\ell}(i)\hat{\theta}^{(\ell)}(u,i)\right) \\
	&= p \left[\sum_{t=1}^T \left(X_t(u)-\sum_{i\in C_u}X_{t-\ell}(i)\hat{\theta}^{(\ell)}(u,i)\right)\right]^2 \\
	& \leq Tp \left[X_t(u)-\sum_{i\in C_u}X_{t-\ell}(i)\hat{\theta}^{(\ell)}(u,i) \right]^2\\
	 &\leq Tp \|\X(u)\|^2_2\\
\end{align*}
The last inequality comes from the Cauchy-Schwarz inequality. Denote by $\nu_{max}$ the largest eigenvalue of the covariance matrix $\boldsymbol\Omega$. Since $\boldsymbol\Omega$ is PSD, we have $(\nu_{max}\mathbf{I} - \boldsymbol\Omega)$ is also PSD. Following \cite{muller2001stochastic}'s argument, we can show $(\hat{\mathbf{R}}_1, \cdots, \hat{\mathbf{R}}_p) \leq_{cx} \mathbf{Y}$ for some random vector $\mathbf{Y} \sim N(\mathbf{0}, \nu_{max}\mathbf{I}_p)$, where $\leq_{cx}$ is the convex order that means $\X \leq \Y$, if and only if $\boldsymbol\mu_x = \boldsymbol\mu_y$ and $\sigma_x^2 \leq \sigma_y^2$. It follows that
\begin{align*}
	\max_{u \in V, v\in V \backslash C_u } \mathbb{P}\left(\left\|G(\hat{\boldsymbol\theta}(u,v))\right\|_2^2 \geq \lambda^2\right) &\leq \max_{u \in V, v\in V \backslash C_u} \mathbb{P}(4T^{-2}(\mathbf{Y}^\prime\mathbf{Y}) \geq \lambda^2)\\
	&= \max_{u \in V, b\in V \backslash C_u} \mathbb{P}\left(\frac{1}{\nu_{max}}\mathbf{Y}^\prime\mathbf{Y} \geq \frac{\lambda^2T^2}{4\nu_{max}}\right) \enskip . \\
\end{align*}
Note that the matrix $\frac{1}{\nu_{max}}\mathbf{Y}^\prime\mathbf{Y}$ is idempotent and thus it follows a $\chi^2(p)$ distribution, and $\nu_{max} \leq \textbf{tr}(\boldsymbol\Omega) \leq Tp\|\X(u)\|_2^2$. Put everything together, we have
\begin{align*}
	\max_{u \in V, b\in V \backslash C_u} \mathbb{P}\left(\|G(\hat{\boldsymbol\theta}(u,v))\|_2^2 \geq \lambda^2\right) &\leq \max_{u \in V, v\in V \backslash C_u} \mathbb{P}\left(\chi^2(p) \geq \frac{\lambda^2T^2}{4\nu_{max}} \right)\\ 
	&\leq \max_{u \in V, v\in V \backslash C_u} \mathbb{P}\left(\chi^2(p) \geq \frac{\lambda^2T^2}{4Tp\|\X(u)\|_2^2} \right) \leq \frac{\alpha}{N(N-1)}
\end{align*}
and thus we have the desired $\lambda(\alpha, a)$
\begin{equation}
	\lambda(\alpha) = 2\hat{\sigma}_u\sqrt{pQ\left(1-\frac{\alpha}{N(N-1)}\right)}.
\end{equation}
\end{proof}
\subsection{Proof of theorem \ref{Riskbound}}
\noindent The proof of the theorem is in line with the work in \cite{kolaczyk2005multiscale}. The core idea is to bound the expected Hellinger loss in terms of the Kullback-Leibler distance.  This approach, building on the original work of~\cite{li2000mixture}, leverages the union of unions bound, after discretizing the underlying parameter space.  We assume a similar discretization here, while omitting the straightforward but tedious numerical analysis arguments that accompany.  See, for example,~\cite{kolaczyk2005multiscale} for details.  Our fundamental bound is given by the following theorem.
\begin{Thm}
Let $\Gamma_T^{(N-1)p}$ be a space of finite collection of estimators $\boldsymbol{\tilde\theta}$ for $\boldsymbol\theta$, and $\text{pen}(\cdot)$ a function on $\Gamma_T^p$ satisfying the condition\\
\begin{align}
\sum_{\boldsymbol{\tilde\theta}(u, v) \in \Gamma_T^p}e^{-pen(\boldsymbol{\tilde\theta}(u, v))} \leq 1,
\label{Inequality}
\end{align}
Let $\hat{\boldsymbol\theta}$ be a penalized maximum likelihood estimator of the form
\begin{align*}
\hat{\boldsymbol\theta} \equiv \argmin_{\boldsymbol{\tilde\theta} \in \Gamma_T^{(N-1)p}}\left\{ -\log p(\X(u)|\X(-u),\boldsymbol{\tilde\theta}) + 2\sum_{v \in V\backslash \{u\}} \text{Pen}(\boldsymbol{\tilde\theta}(u,v))\right\}.
\end{align*}
Then
\begin{align}
\mathbb{E}[H^2(p_{\hat{\boldsymbol\theta}},p_{\boldsymbol\theta})] \leq \min_{\boldsymbol{\tilde\theta} \in \Gamma_T^{(N-1)p}}\left\{K(p_{\boldsymbol\theta},p_{\boldsymbol{\tilde\theta}}) + 2\sum_{v \in V\backslash \{u\}}\text{Pen}(\boldsymbol{\tilde\theta}(u,v))\right\}.
\label{bounded.expectation}
\end{align}
\label{Thm.Partition}
\end{Thm}
Note that the result of theorem \ref{Thm.Partition} requires that inequality (\ref{Inequality}) holds. Lemma \ref{Lemma.kraft} shows that our proposed penalty satisfies inequality (\ref{Inequality}). We now prove theorem \ref{Thm.Partition}.

\begin{proof}
Note that we have
\begin{align*}
H^2(p_{\hat{\boldsymbol\theta}},p_{\boldsymbol\theta}) &= \int \left[\sqrt{p(\x|\XX, \boldsymbol{\hat\theta})} - \sqrt{p(\x|\XX, \boldsymbol\theta)}\right]^2 d\nu(\x) \\
&= 2\left(1 - \int \sqrt{p(\x|\XX, \hat{\boldsymbol\theta})p(\x|\XX, \boldsymbol\theta)}d\nu(\x) \right)\\ 
&\leq -2 \log \int \sqrt{p(\x|\XX, \hat{\boldsymbol\theta})p(\x|\XX, \boldsymbol\theta)} d\nu(\x),
\end{align*}
Taking the conditional expectation respect to $\X(u) |\XX$, we then have
\small
\begin{align*}
\mathbb{E}[H^2(p_{\hat{\boldsymbol\theta}},p_{\boldsymbol\theta})] &\leq 2\mathbb{E}\log \left(\frac{1}{\int \sqrt{p(\x|\XX,\hat{\boldsymbol\theta})p(\x|\XX,\boldsymbol\theta)}d\nu(\x)} \right) \\
& \leq 2\mathbb{E}\log \left(\frac{p^{1/2}(\X(u)|\XX,\hat{\boldsymbol\theta})e^{- \sum\limits_v pen(\hat{\boldsymbol\theta}(u,v))}}{p^{1/2}(\X(u)|\XX,\check{\boldsymbol\theta})e^{- \sum\limits_v pen(\check{\boldsymbol\theta}(u,v))}} \frac{1}{\int \sqrt{p(\x|\XX,\hat{\boldsymbol\theta})p(\x|\XX,\boldsymbol\theta)}d\nu(\x)} \right),
\end{align*}
\normalfont
where the collection of $\check{\boldsymbol\theta}(u,v)$'s are the arguments that minimize the right-hand side of the expression (\ref{bounded.expectation}). The last expression can be written in two pieces, that is
\begin{align}
& \mathbb{E}\left[ \log \frac{p(\X(u)|\XX,\boldsymbol\theta)}{p(\X(u)|\XX,\check{\boldsymbol\theta})}\right] + 2 \sum_{v} pen(\check{\boldsymbol\theta}(u,v)) 
 \label{part1}\\
&+ 2\mathbb{E} \log \left(\frac{p^{1/2}(\X(u)|\XX,\hat{\boldsymbol\theta})}{p^{1/2}(\X(u)|\XX,\boldsymbol\theta)}\frac{\prod\limits_v \prod\limits_\ell e^{-pen(\hat{\boldsymbol\theta}^{(\ell)}(u,v))}}{\int \sqrt{p(\x|\XX,\hat{\boldsymbol\theta})p(\x|\XX,\boldsymbol\theta)}d\nu(\x)} \right)
 \label{part2}
\end{align}
Note that the expression (\ref{part1}) is the right hand side of (\ref{bounded.expectation}). What we need to show then is that expression (\ref{part2}) is bounded above by zero. By applying Jensen's inequality, we have (\ref{part2}) bounded by:
\begin{align}
2\log \mathbb{E}\left[\prod_{v}e^{-pen(\hat{\boldsymbol\theta}(u,v))}\frac{\sqrt{p(\X(u)|\XX,\hat{\boldsymbol\theta})/p(\X(u)|\XX,\boldsymbol\theta)}}{\int \sqrt{p(\x|\XX,\hat{\boldsymbol\theta})p(\x|\XX,\boldsymbol\theta)}d\nu(\x)} \right]
\label{upperbound}
\end{align}
The integrand in the expectation in (\ref{upperbound}) can be bounded by
\begin{align*}
\sum_{\boldsymbol{\tilde\theta} \in \Gamma_T^{(N-1)p}} \prod_{v}e^{-pen(\tilde{\boldsymbol\theta}(u,v))}\frac{\sqrt{p(\X(u)|\XX,\tilde{\boldsymbol\theta})/p(\X(u)|\XX,\boldsymbol\theta)}}{\int \sqrt{p(\x|\XX,\tilde{\boldsymbol\theta})p(\x|\XX,\boldsymbol\theta)}d\nu(\x)}.
\end{align*}
Given the fact that $\tilde{\boldsymbol\theta}$ does not depend on the $\X(-u)$, (\ref{upperbound}) can be bounded by
\begin{align}
& 2\log \sum_{\boldsymbol{\tilde\theta} \in \Gamma_T^{(N-1)p}} \prod_{v}e^{-pen(\tilde{\boldsymbol\theta}(u,v))}\frac{\mathbb{E}\left[\sqrt{p(\X(u)|\XX,\tilde{\boldsymbol\theta})/p(\X(u)|\XX,\boldsymbol\theta)}\right]}{\int \sqrt{p(\x|\XX,\tilde{\boldsymbol\theta})p(\x|\XX,\boldsymbol\theta)}d\nu(\x)}\nonumber \\
=\,& 2\log \sum_{\boldsymbol{\tilde\theta} \in \Gamma_T^{(N-1)p}} \prod_{v}e^{-pen(\tilde{\boldsymbol\theta}(u,v))} 
\label{prodofsum}
\end{align}
Since $e^{-pen(\tilde{\boldsymbol\theta}(u,v))} > 0$  for any $\boldsymbol{\tilde\theta}(u,v)$, and using the inequality $\sum_i a_i b_i \leq \sum_i a_i \sum_ib_i$ for any $a_i > 0, b_i>0$, we can bound (\ref{prodofsum}) by:
\begin{align*}
2\log \prod_{v} \sum_{\boldsymbol{\tilde\theta}(u,v) \in \Gamma_T^p} e^{- pen(\tilde{\boldsymbol\theta}(u,v))}
\end{align*}
\end{proof}
From the condition in (\ref{Inequality}), we see that the above expression is bounded by zero. We now show that our proposed estimator satisfies condition (\ref{Inequality}) by the following lemma. 
\begin{Lem}
Let $\Gamma_T$ be the collection of all $\boldsymbol{\tilde\theta}^{(\ell)}(u, v)$ with components $\boldsymbol{\tilde\theta}_t^{(\ell)}(u, v) \in D_T[-C, C]$ and possessed of a Haar like expansion through a common partition, using either RDP (see expression (\ref{Space_RDP})) or RP (see expression (\ref{Space_RP})), where $D_T[-C, C]$ denotes a discretization of the interval $[-C, C]$ into $T^{1/2}$ equispaced values. For any type of penalty such that
\begin{align*}
Pen(\boldsymbol{\tilde\theta}(u, v)) = C_3\log T \#\{\PP(\boldsymbol{\tilde\theta)}\} + \lambda\sum_{\mathcal{I} \in \PP(\boldsymbol{\tilde\theta})} \|\boldsymbol{\tilde\theta}_{\mathcal{I}}(u, v)\|_2,
\end{align*}
where $C_3=1/2$ for recursive dyadic partitioning and $C_3=3/2$ for recursive partitioning,
we have
\begin{align*}
\sum_{\boldsymbol{\tilde\theta}(u, v) \in \Gamma_T^p} e^{-pen(\boldsymbol{\boldsymbol{\tilde\theta}}(u, v))} \leq 1, 
\end{align*}
for $T > \lceil e^{2p/3} \rceil$.
\label{Lemma.kraft}
\end{Lem}
\begin{proof}
We prove Lemma \ref{Lemma.kraft} for the case of recursive partitioning. We write $\Gamma_T = \bigcup_{d_\ell=1}^T \Gamma_T^{(d_\ell)}$ where $\Gamma_T^{(d_\ell)}$ is the subset of values $\boldsymbol{\tilde\theta}_t^{(\ell)}(u, v)$ that is composed of $d_\ell$ constant valued sequences. For example, $\Gamma_T^{(d_\ell)}$ consists of all length $T$ sequences such that there are exactly $d_\ell$ alternating sequences of zero and nonzero elements. So, for example, $(0, 0, 4, 0, 0)$ and $(2, 0, 1, 1, 1)$ might be two such sequences in $\Gamma_5^{(3)}$. Then we have
\begin{align*}
\sum_{\boldsymbol{\tilde\theta}(u, v) \in \Gamma_T^p} e^{-pen(\boldsymbol{\tilde\theta}(u, v))} &= \sum_{\boldsymbol{\tilde\theta}(u, v) \in \Gamma_T^p} e^{-(3/2)\log T\{\#\PP(\boldsymbol{\tilde\theta)}\} - \lambda\sum\limits_{\mathcal{I}\in \PP(\boldsymbol{\tilde\theta})}\|\boldsymbol{\tilde\theta}_\mathcal{I}(u, v) \|_2} \\
& \leq \sum_{\boldsymbol{\tilde\theta}(u, v) \in \Gamma_T^p} e^{-(3/2)\log T\{\#\PP(\boldsymbol{\tilde\theta)}\}}\\
& \leq \prod_{\ell=1}^p  \sum_{\boldsymbol{\tilde\theta}^{(\ell)}(u, v) \in \Gamma_T} e^{-(3/2p)\log T\{\#\PP(\boldsymbol{\tilde\theta)}\}}\\
& = \prod_{\ell=1}^p \sum_{d_\ell=1}^T\binom{T-1}{d_\ell-1} e^{-d_\ell(3/2p)\log T}\\
& = \prod_{\ell=1}^p  \sum_{d^{\ell^\prime}=0}^{T-1}\binom{T-1}{d^{\ell^\prime}}e^{-(d^{\ell^\prime}+1)(3/2p)\log T} \\
& = \prod_{\ell=1}^p  \sum_{d\prime=0}^{T-1} \frac{(T-1)!}{d^{\ell^\prime} ! (T-d^{\ell^\prime}-1)!} T^{-(d^{\ell^\prime}+1)(3/2p)}\\
& \leq \prod_{\ell=1}^p T^{-(3/2p)} \sum_{d^{\ell^\prime}=0}^{T-1} \frac{(T-1)^{d^{\ell^\prime}}}{d^{\ell^\prime}!}\frac{1}{T^{(3/2p)d^{\ell^\prime}}}\\
& \leq T^{-(3/2)} e^p
\end{align*}
which is bounded by $1$ for any $T > \lceil e^{2p/3} \rceil$. The argument follows analogously for the case of  recursive dyadic partitioning.
\end{proof}

Using the loss function and the corresponding risk function we defined before, recovering the neighborhood of node $u$ is essentially a univariate Gaussian time series problem, and thus the KL divergence of the conditional likelihood function takes the form:
\begin{align*}
K(p_{\boldsymbol\theta},p_{\boldsymbol{\tilde\theta}}) = \mathbb{E}\left\{\log \frac{p_{\boldsymbol\theta}(\x)}{p_{\boldsymbol{\tilde\theta}}(\x)}\right\} = \mathbb{E}\left\{\sum_{t=1}^T \log \frac{p_{\boldsymbol\theta}(X_t(u))}{p_{\boldsymbol{\tilde\theta}}(X_t(u))}\right\} = \sum_{t=1}^T (\tilde\mu_t-\mu_t)^2 / (2\sigma^2)
\end{align*}
where each $\mu_t$ is the mean of $X_t(u)$, and $\tilde\mu_t$ is an approximation/estimate thereof, for a given estimator $\boldsymbol{\tilde\theta}$. Since these means in turn are based on linear combinations of all neighborhood observations, over $p$ lags, we have:
\begin{align*}
\tilde\mu_t - \mu_t = \sum_{v\in V \backslash \{u\}}\sum_{\ell=1}^p  X_{t-\ell}(v)[\tilde\theta_t^{(\ell)}(u, v) - \theta_t^{(\ell)}(u, v)]
\end{align*}
So the KL divergence for each neighborhood problem involves values at other nodes.

Assume without loss of generality that $\sigma \equiv 1$. From (\ref{bounded.expectation}) and the fact that the K-L divergence in the Gaussian case is simply proportional to a squared $\ell_2$-norm, the risk of estimating $\boldsymbol{\theta}$ by $\boldsymbol{\hat\theta}$ should be in the form:
{\small
\begin{align*}
\mathbb{R}(\hat{\boldsymbol\theta}, \boldsymbol\theta) &\leq \min_{\boldsymbol{\tilde\theta} \in \Gamma_T^{(N-1)p}}\left\{\frac{1}{T}K(p_{\boldsymbol\theta},p_{\boldsymbol{\tilde\theta}}) + \frac{2}{T}\sum_{v=1}^{N-1}  Pen(\boldsymbol{\tilde\theta}(u,v))\right\} \\
&\leq \min_{\boldsymbol{\tilde\theta} \in \Gamma_T^{(N-1)p}} \left\{\frac{1}{2T}\left\|\boldsymbol{\tilde\mu}-\boldsymbol\mu\right\|_2^2 + \frac{\lambda}{T}\sum_{\mathcal{I} \in \PP(\boldsymbol{\tilde\theta})}\sum_{v=1}^{N-1}\|\boldsymbol{\tilde\theta}_\mathcal{I}(u,v)\|_2 + \frac{2}{T}\sum_{v=1}^{N-1}(3/2)\log T \#\{\PP(\boldsymbol{\tilde\theta})\}\right\}\\
\end{align*}}

From Cauchy-Schwarz, we have that
{\small
\begin{align}
\mathbb{R}(\hat{\boldsymbol\mu}, \boldsymbol\mu) 
&\leq  \min_{\boldsymbol{\tilde\theta} \in \Gamma_T^{(N-1)p}}\left\{\frac{1}{2T}\|{\X(-u)}^\prime\X(-u)\|_2 \sum_{t=1}^T\sum_{v=1}^{N-1}\sum_{\ell=1}^p\left(\tilde\theta_t^{(\ell)}(u,v) - \theta_t^{(\ell)}(u,v)\right)^2 \right.\nonumber\\
&\quad\quad +\left.
 \frac{\lambda}{T}\sum_{\mathcal{I} \in \PP(\boldsymbol{\tilde\theta})}\sum_{v=1}^{N-1}\|\boldsymbol{\tilde\theta}_{\mathcal{I}}(u,v)\|_2 + 3(N-1)\frac{\log T}{T}\#\{\PP(\boldsymbol{\tilde\theta})\}\right\} \nonumber\\
&\leq \min_{\boldsymbol{\tilde\theta} \in \Gamma_T^{(N-1)p}}\left\{\frac{1}{2}\Lambda \sum_{v=1}^{N-1}\sum_{\ell=1}^p\left\|\boldsymbol{\tilde\theta}_t^{(\ell)}(u,v) - \boldsymbol{\theta}_t^{(\ell)}(u,v)\right\|_2^2 \right.\nonumber\\
&\quad\quad \left. + \frac{\lambda}{T}\sum_{\mathcal{I} \in \PP(\boldsymbol{\tilde\theta})}\sum_{v=1}^{N-1}\|\boldsymbol{\tilde\theta}_{\mathcal{I}}(u,v)\|_2 + 3(N-1)\frac{\log T}{T}\#\{\PP(\boldsymbol{\tilde\theta})\}\right\} \enskip .
\label{bound}
\end{align}

The minimization of the expression (\ref{bound}) tries to find the optimal balancing of bias and variance. To bound it, the following $L_2$ result from \cite{donoho1993unconditional} plays the core role.
 
\begin{Lem}\label{Lemma.donoho}
Let $\theta_{(\cdot)}^{(\ell)}(u,v) \in BV(C)$. Define ${\theta_{bd}}_{(\cdot)}^{(\ell)}(u,v)$ to be the best $d$-term approximant to $\theta_{(\cdot)}^{(\ell)}(u,v)$ in the dyadic Haar basis for $L_2([0, 1])$. Then $\|{\theta_{bd}}^{(\ell)}(u,v) - \theta^{(\ell)}(u,v)\|_{L_2} = \mathcal{O}(d^{-1})$.
\end{Lem}

Define ${\boldsymbol{\theta}_{bd}}^{(\ell)}(u,v)$ to be the average sampling of ${\theta_{bd}}^{(\ell)}(u,v)$ on the interval $I_i$, that is ${\boldsymbol{\theta}_{bd}}^{(\ell)}(u,v) = T\int_{I_i}{\theta_{bd}}^{(\ell)}(u,v)(t) dt$. Then let ${\boldsymbol{\tilde\theta}_{bd}}^{(\ell)}(u,v)$ be the result of discretizing the elements of ${\boldsymbol{\theta}_{bd}}^{(\ell)}(u,v)$ to the set $D_T[-C, C]$, where $C$ is the radius of the bounded variation ball defined in Assumption \ref{A6}. We have the following by triangle inequality:
\begin{align}
\left\|\boldsymbol{\tilde\theta}^{(\ell)}(u,v) - \boldsymbol{\theta}^{(\ell)}(u,v)\right\|_{\ell_2}^2 
&\leq  \left\|{\boldsymbol{\theta}_{bd}}^{(\ell)}(u,v) - \boldsymbol{\theta}^{(\ell)}(u,v)\right\|_{\ell_2}^2 + \left\|\boldsymbol{\tilde\theta}^{(\ell)}(u,v) - {\boldsymbol{\theta}_{bd}}^{(\ell)}(u,v)\right\|_{\ell_2}^2 \nonumber\\
&+ 2\left\|{\boldsymbol{\theta}_{bd}}^{(\ell)}(u,v) - \boldsymbol{\theta}^{(\ell)}(u,v)\right\|_{\ell_2}\left\|\boldsymbol{\tilde\theta}^{(\ell)}(u,v) - \boldsymbol{\theta}^{(\ell)}(u,v)\right\|_{\ell_2}.
\label{complete.approx}
\end{align}
For sequence ${\boldsymbol{\theta}_{bd}}^{(\ell)}(u,v)$ and ${\boldsymbol{\tilde\theta}_{bd}}^{(\ell)}(u,v)$ obtained from average sampling, a simple argument relating Haar function on the discrete set $D_T[-C, C]$ to the functions on the interval $[0, 1]$ is to show that 
\begin{align*}
\frac{1}{T}\left\|{\boldsymbol{\tilde\theta}_{bd}}^{(\ell)}(u,v) - \boldsymbol{\theta}^{(\ell)}(u,v)\right\|_{\ell_2}^2 \leq \left\|\theta_{bd}^{(\ell)}(u,v) - \theta^{(\ell)}(u,v) \right\|_{L2}^2 \enskip .
\end{align*}
See equation (27) of \cite{kolaczyk2005multiscale}. On the right hand side of (\ref{complete.approx}), the first resulting squared term will be of order $\mathcal{O}(Td^{-2})$. The second term is a discretization error and by lemma (\ref{Lemma.donoho}) is of order $\mathcal{O}(1)$. The third cross-term is therefore of order $\mathcal{O}(T^{1/2}d^{-1})$.

Given these results, we have the following bound of equation (\ref{bound}) by bounding the bias term over each $\Gamma_T^{(d)}$, where $d = \bigcup_i d_i$, for each $d_i$ and $i = 1,\cdots, (N-1)p$. We then we optimize for $d$:
{\small
\begin{align}
&\min_{\boldsymbol{\tilde\theta} \in {\Gamma_T^{(N-1)p}}^{(d)}}
\left\{\frac{1}{2}\Lambda \sum_{v=1}^{N-1}\sum_{\ell=1}^p\left\|\boldsymbol{\tilde\theta}^{(\ell)}(u,v) - \boldsymbol\theta^{(\ell)}(u,v)\right\|_2^2 \right.\nonumber\\
&\quad\quad \left. + \frac{\lambda}{T}\sum_{\mathcal{I} \in \PP(\boldsymbol{\tilde\theta})}\sum_{v=1}^{N-1}\|\boldsymbol{\tilde\theta}_{\mathcal{I}}(u,v)\|_2 + 3(N-1)\frac{\log T}{T}\#\{\PP(\boldsymbol{\tilde\theta})\}\right\}
\label{discrete.bound}
\end{align}
}
The first term is dominated by the first part of expression (\ref{complete.approx}) and is of order $\mathcal{O}(\Lambda Td^{-2})$. In the second term, we have $\frac{\lambda}{T}\sum_{\mathcal{I} \in \PP(\boldsymbol{\tilde\theta})}\sum_{v=1}^{N-1}\|\boldsymbol{\tilde\theta}_{\mathcal{I}}(u,v)\|_2$, which are the group lasso terms. Given the fact that $\theta_{(\cdot)}^{(\ell)}(u,v)$ is of $BV(C)$, we have that $1/(T^{1/2})\|\boldsymbol{\tilde\theta}_{\mathcal{I}}(u,v)\|_2$ is of order $\mathcal{O}(C + d^{-1})$. Note that $\lambda$ is of order $T^{-1/2}$ and the number of interval $\#\{\PP(\boldsymbol{\tilde\theta})\}$ is proportional to $d$. So the second term is of order $\mathcal{O}(T^{-1} *d * (C + d^{-1}) )$. The third term is of order $\mathcal{O}(dT^{-1}\log T)$. Combining the above results, we have that: 
\begin{align*}
& \min_{\boldsymbol{\tilde\theta} \in {\Gamma_T^{(N-1)p}}^{(d)}}
\left\{\frac{1}{2}\Lambda \sum_{v=1}^{N-1}\sum_{\ell=1}^p\left\|\boldsymbol{\tilde\theta}^{(\ell)}(u,v) - \boldsymbol\theta^{(\ell)}(u,v)\right\|_2^2 \right.\nonumber\\
&\quad\quad \left. + \frac{\lambda}{T}\sum_{\mathcal{I} \in \PP(\boldsymbol{\tilde\theta})}\sum_{v=1}^{N-1}\|\boldsymbol{\tilde\theta}_{\mathcal{I}}(u,v)\|_2 + 3(N-1)\frac{\log T}{T}\#\{\PP(\boldsymbol{\tilde\theta})\}\right\}\\
&\leq \mathcal{O}(\Lambda Td^{-2}) + \mathcal{O}(T^{-1} *d * (C + d^{-1}) ) + \mathcal{O}(dT^{-1}\log T),
\end{align*}
which is minimized for $d \sim (\Lambda T^2/\log T)^{1/3}$. Substitution then yields the result that the risk is bounded by a quantity of order $\mathcal{O}((\Lambda\log ^2T/T)^{1/3})$. For estimation via recursive dyadic partitioning, where $\#\{\PP(\tilde\theta)\}$ is proportional to $d\log T$, the expression is minimized at $d \sim (\Lambda T^2 /\log^2 T)^{1/3}$, which gives the bound of the risk of order $\mathcal{O}(\Lambda \log^4T/T)^{1/3}$.
\end{document}